\documentclass[12pt]{article}	
\usepackage{amsmath, amssymb, amsmath, amsthm} 
\usepackage{bm}														
\usepackage[usenames,dvipsnames]{color}								
\usepackage{setspace}	

\usepackage{geometry}												
\usepackage[sc]{mathpazo}
\usepackage[T1]{fontenc}
\usepackage[utf8]{inputenc}
\usepackage[english]{babel}

\usepackage{graphicx}												

\usepackage{enumerate}
\usepackage{url}													
\usepackage{natbib}													

\usepackage[T1]{fontenc}
\usepackage[plainpages=false]{hyperref}								
		
\usepackage[usenames,dvipsnames,svgnames,table]{xcolor}
\hypersetup{
 colorlinks=true,
 citecolor=Maroon,	
 linkcolor=Maroon,
 urlcolor=Maroon}																						
\usepackage{pst-all}                            					
\usepackage{nicefrac}												
\usepackage{booktabs}
\usepackage{float}

\usepackage{epsfig} 

\newcommand{\E}{\mathbb{E}}

\newcommand{\p}{\partial}
\newcommand{\Var}{\ensuremath{\textrm{Var}}}

\newtheorem{Assumption}{Assumption}

\newtheorem{Example}{Example}

\newtheorem{Corollary}{Corollary}

\newtheorem{Lemma}{Lemma}
\newtheorem*{Proof*}{Proof}
\newtheorem{Proposition}{Proposition}

\definecolor{snow}{rgb}{1.000,0.980,0.980}
\definecolor{snow1}{rgb}{1.000,0.980,0.980}
\definecolor{snow2}{rgb}{0.933,0.914,0.914}
\definecolor{snow3}{rgb}{0.804,0.788,0.788}
\definecolor{snow4}{rgb}{0.545,0.537,0.537}
\definecolor{GhostWhite}{rgb}{0.973,0.973,1.000}
\definecolor{WhiteSmoke}{rgb}{0.961,0.961,0.961}
\definecolor{gainsboro}{rgb}{0.863,0.863,0.863}
\definecolor{FloralWhite}{rgb}{1.000,0.980,0.941}
\definecolor{OldLace}{rgb}{0.992,0.961,0.902}
\definecolor{linen}{rgb}{0.980,0.941,0.902}
\definecolor{AntiqueWhite}{rgb}{0.980,0.922,0.843}
\definecolor{PapayaWhip}{rgb}{1.000,0.937,0.835}
\definecolor{BlanchedAlmond}{rgb}{1.000,0.922,0.804}
\definecolor{bisque}{rgb}{1.000,0.894,0.769}
\definecolor{PeachPuff}{rgb}{1.000,0.855,0.725}
\definecolor{NavajoWhite}{rgb}{1.000,0.871,0.678}
\definecolor{moccasin}{rgb}{1.000,0.894,0.710}
\definecolor{cornsilk}{rgb}{1.000,0.973,0.863}
\definecolor{ivory}{rgb}{1.000,1.000,0.941}
\definecolor{LemonChiffon}{rgb}{1.000,0.980,0.804}
\definecolor{seashell}{rgb}{1.000,0.961,0.933}
\definecolor{honeydew}{rgb}{0.941,1.000,0.941}
\definecolor{MintCream}{rgb}{0.961,1.000,0.980}
\definecolor{azure}{rgb}{0.941,1.000,1.000}
\definecolor{AliceBlue}{rgb}{0.941,0.973,1.000}
\definecolor{lavender}{rgb}{0.902,0.902,0.980}
\definecolor{LavenderBlush}{rgb}{1.000,0.941,0.961}
\definecolor{MistyRose}{rgb}{1.000,0.894,0.882}
\definecolor{white}{rgb}{1.000,1.000,1.000}
\definecolor{black}{rgb}{0.000,0.000,0.000}
\definecolor{DarkSlateGray}{rgb}{0.184,0.310,0.310}
\definecolor{DimGray}{rgb}{0.412,0.412,0.412}
\definecolor{SlateGray}{rgb}{0.439,0.502,0.565}
\definecolor{LightSlateGray}{rgb}{0.467,0.533,0.600}
\definecolor{gray}{rgb}{0.745,0.745,0.745}
\definecolor{LightGray}{rgb}{0.827,0.827,0.827}
\definecolor{MidnightBlue}{rgb}{0.098,0.098,0.439}
\definecolor{navy}{rgb}{0.000,0.000,0.502}
\definecolor{NavyBlue}{rgb}{0.000,0.000,0.502}
\definecolor{CornflowerBlue}{rgb}{0.392,0.584,0.929}
\definecolor{DarkSlateBlue}{rgb}{0.282,0.239,0.545}
\definecolor{SlateBlue}{rgb}{0.416,0.353,0.804}
\definecolor{MediumSlateBlue}{rgb}{0.482,0.408,0.933}
\definecolor{LightSlateBlue}{rgb}{0.518,0.439,1.000}
\definecolor{MediumBlue}{rgb}{0.000,0.000,0.804}
\definecolor{RoyalBlue}{rgb}{0.255,0.412,0.882}
\definecolor{blue}{rgb}{0.000,0.000,1.000}
\definecolor{DodgerBlue}{rgb}{0.118,0.565,1.000}
\definecolor{DeepSkyBlue}{rgb}{0.000,0.749,1.000}
\definecolor{SkyBlue}{rgb}{0.529,0.808,0.922}
\definecolor{LightSkyBlue}{rgb}{0.529,0.808,0.980}
\definecolor{SteelBlue}{rgb}{0.275,0.510,0.706}
\definecolor{LightSteelBlue}{rgb}{0.690,0.769,0.871}
\definecolor{LightBlue}{rgb}{0.678,0.847,0.902}
\definecolor{PowderBlue}{rgb}{0.690,0.878,0.902}
\definecolor{PaleTurquoise}{rgb}{0.686,0.933,0.933}
\definecolor{DarkTurquoise}{rgb}{0.000,0.808,0.820}
\definecolor{MediumTurquoise}{rgb}{0.282,0.820,0.800}
\definecolor{turquoise}{rgb}{0.251,0.878,0.816}
\definecolor{cyan}{rgb}{0.000,1.000,1.000}
\definecolor{LightCyan}{rgb}{0.878,1.000,1.000}
\definecolor{CadetBlue}{rgb}{0.373,0.620,0.627}
\definecolor{MediumAquamarine}{rgb}{0.400,0.804,0.667}
\definecolor{aquamarine}{rgb}{0.498,1.000,0.831}
\definecolor{DarkGreen}{rgb}{0.000,0.392,0.000}
\definecolor{DarkOliveGreen}{rgb}{0.333,0.420,0.184}
\definecolor{DarkSeaGreen}{rgb}{0.561,0.737,0.561}
\definecolor{SeaGreen}{rgb}{0.180,0.545,0.341}
\definecolor{MediumSeaGreen}{rgb}{0.235,0.702,0.443}
\definecolor{LightSeaGreen}{rgb}{0.125,0.698,0.667}
\definecolor{PaleGreen}{rgb}{0.596,0.984,0.596}
\definecolor{SpringGreen}{rgb}{0.000,1.000,0.498}
\definecolor{LawnGreen}{rgb}{0.486,0.988,0.000}
\definecolor{green}{rgb}{0.000,1.000,0.000}
\definecolor{chartreuse}{rgb}{0.498,1.000,0.000}
\definecolor{MediumSpringGreen}{rgb}{0.000,0.980,0.604}
\definecolor{GreenYellow}{rgb}{0.678,1.000,0.184}
\definecolor{LimeGreen}{rgb}{0.196,0.804,0.196}
\definecolor{YellowGreen}{rgb}{0.604,0.804,0.196}
\definecolor{ForestGreen}{rgb}{0.133,0.545,0.133}
\definecolor{OliveDrab}{rgb}{0.420,0.557,0.137}
\definecolor{DarkKhaki}{rgb}{0.741,0.718,0.420}
\definecolor{khaki}{rgb}{0.941,0.902,0.549}
\definecolor{PaleGoldenrod}{rgb}{0.933,0.910,0.667}
\definecolor{LightGoldenrodYellow}{rgb}{0.980,0.980,0.824}
\definecolor{LightYellow}{rgb}{1.000,1.000,0.878}
\definecolor{yellow}{rgb}{1.000,1.000,0.000}
\definecolor{gold}{rgb}{1.000,0.843,0.000}
\definecolor{LightGoldenrod}{rgb}{0.933,0.867,0.510}
\definecolor{goldenrod}{rgb}{0.855,0.647,0.125}
\definecolor{DarkGoldenrod}{rgb}{0.722,0.525,0.043}
\definecolor{RosyBrown}{rgb}{0.737,0.561,0.561}
\definecolor{IndianRed}{rgb}{0.804,0.361,0.361}
\definecolor{SaddleBrown}{rgb}{0.545,0.271,0.075}
\definecolor{sienna}{rgb}{0.627,0.322,0.176}
\definecolor{peru}{rgb}{0.804,0.522,0.247}
\definecolor{burlywood}{rgb}{0.871,0.722,0.529}
\definecolor{beige}{rgb}{0.961,0.961,0.863}
\definecolor{wheat}{rgb}{0.961,0.871,0.702}
\definecolor{SandyBrown}{rgb}{0.957,0.643,0.376}
\definecolor{tan}{rgb}{0.824,0.706,0.549}
\definecolor{chocolate}{rgb}{0.824,0.412,0.118}
\definecolor{firebrick}{rgb}{0.698,0.133,0.133}
\definecolor{brown}{rgb}{0.647,0.165,0.165}
\definecolor{DarkSalmon}{rgb}{0.914,0.588,0.478}
\definecolor{salmon}{rgb}{0.980,0.502,0.447}
\definecolor{LightSalmon}{rgb}{1.000,0.627,0.478}
\definecolor{orange}{rgb}{1.000,0.647,0.000}
\definecolor{DarkOrange}{rgb}{1.000,0.549,0.000}
\definecolor{coral}{rgb}{1.000,0.498,0.314}
\definecolor{LightCoral}{rgb}{0.941,0.502,0.502}
\definecolor{tomato}{rgb}{1.000,0.388,0.278}
\definecolor{OrangeRed}{rgb}{1.000,0.271,0.000}
\definecolor{red}{rgb}{1.000,0.000,0.000}
\definecolor{HotPink}{rgb}{1.000,0.412,0.706}
\definecolor{DeepPink}{rgb}{1.000,0.078,0.576}
\definecolor{pink}{rgb}{1.000,0.753,0.796}
\definecolor{LightPink}{rgb}{1.000,0.714,0.757}
\definecolor{PaleVioletRed}{rgb}{0.859,0.439,0.576}
\definecolor{maroon}{rgb}{0.690,0.188,0.376}
\definecolor{MediumVioletRed}{rgb}{0.780,0.082,0.522}
\definecolor{VioletRed}{rgb}{0.816,0.125,0.565}
\definecolor{magenta}{rgb}{1.000,0.000,1.000}
\definecolor{violet}{rgb}{0.933,0.510,0.933}
\definecolor{plum}{rgb}{0.867,0.627,0.867}
\definecolor{orchid}{rgb}{0.855,0.439,0.839}
\definecolor{MediumOrchid}{rgb}{0.729,0.333,0.827}
\definecolor{DarkOrchid}{rgb}{0.600,0.196,0.800}
\definecolor{DarkViolet}{rgb}{0.580,0.000,0.827}
\definecolor{BlueViolet}{rgb}{0.541,0.169,0.886}
\definecolor{purple}{rgb}{0.627,0.125,0.941}
\definecolor{MediumPurple}{rgb}{0.576,0.439,0.859}
\definecolor{thistle}{rgb}{0.847,0.749,0.847}
\definecolor{seashell1}{rgb}{1.000,0.961,0.933}
\definecolor{seashell2}{rgb}{0.933,0.898,0.871}
\definecolor{seashell3}{rgb}{0.804,0.773,0.749}
\definecolor{seashell4}{rgb}{0.545,0.525,0.510}
\definecolor{AntiqueWhite1}{rgb}{1.000,0.937,0.859}
\definecolor{AntiqueWhite2}{rgb}{0.933,0.875,0.800}
\definecolor{AntiqueWhite3}{rgb}{0.804,0.753,0.690}
\definecolor{AntiqueWhite4}{rgb}{0.545,0.514,0.471}
\definecolor{bisque1}{rgb}{1.000,0.894,0.769}
\definecolor{bisque2}{rgb}{0.933,0.835,0.718}
\definecolor{bisque3}{rgb}{0.804,0.718,0.620}
\definecolor{bisque4}{rgb}{0.545,0.490,0.420}
\definecolor{PeachPuff1}{rgb}{1.000,0.855,0.725}
\definecolor{PeachPuff2}{rgb}{0.933,0.796,0.678}
\definecolor{PeachPuff3}{rgb}{0.804,0.686,0.584}
\definecolor{PeachPuff4}{rgb}{0.545,0.467,0.396}
\definecolor{NavajoWhite1}{rgb}{1.000,0.871,0.678}
\definecolor{NavajoWhite2}{rgb}{0.933,0.812,0.631}
\definecolor{NavajoWhite3}{rgb}{0.804,0.702,0.545}
\definecolor{NavajoWhite4}{rgb}{0.545,0.475,0.369}
\definecolor{LemonChiffon1}{rgb}{1.000,0.980,0.804}
\definecolor{LemonChiffon2}{rgb}{0.933,0.914,0.749}
\definecolor{LemonChiffon3}{rgb}{0.804,0.788,0.647}
\definecolor{LemonChiffon4}{rgb}{0.545,0.537,0.439}
\definecolor{cornsilk1}{rgb}{1.000,0.973,0.863}
\definecolor{cornsilk2}{rgb}{0.933,0.910,0.804}
\definecolor{cornsilk3}{rgb}{0.804,0.784,0.694}
\definecolor{cornsilk4}{rgb}{0.545,0.533,0.471}
\definecolor{ivory1}{rgb}{1.000,1.000,0.941}
\definecolor{ivory2}{rgb}{0.933,0.933,0.878}
\definecolor{ivory3}{rgb}{0.804,0.804,0.757}
\definecolor{ivory4}{rgb}{0.545,0.545,0.514}
\definecolor{honeydew1}{rgb}{0.941,1.000,0.941}
\definecolor{honeydew2}{rgb}{0.878,0.933,0.878}
\definecolor{honeydew3}{rgb}{0.757,0.804,0.757}
\definecolor{honeydew4}{rgb}{0.514,0.545,0.514}
\definecolor{LavenderBlush1}{rgb}{1.000,0.941,0.961}
\definecolor{LavenderBlush2}{rgb}{0.933,0.878,0.898}
\definecolor{LavenderBlush3}{rgb}{0.804,0.757,0.773}
\definecolor{LavenderBlush4}{rgb}{0.545,0.514,0.525}
\definecolor{MistyRose1}{rgb}{1.000,0.894,0.882}
\definecolor{MistyRose2}{rgb}{0.933,0.835,0.824}
\definecolor{MistyRose3}{rgb}{0.804,0.718,0.710}
\definecolor{MistyRose4}{rgb}{0.545,0.490,0.482}
\definecolor{azure1}{rgb}{0.941,1.000,1.000}
\definecolor{azure2}{rgb}{0.878,0.933,0.933}
\definecolor{azure3}{rgb}{0.757,0.804,0.804}
\definecolor{azure4}{rgb}{0.514,0.545,0.545}
\definecolor{SlateBlue1}{rgb}{0.514,0.435,1.000}
\definecolor{SlateBlue2}{rgb}{0.478,0.404,0.933}
\definecolor{SlateBlue3}{rgb}{0.412,0.349,0.804}
\definecolor{SlateBlue4}{rgb}{0.278,0.235,0.545}
\definecolor{RoyalBlue1}{rgb}{0.282,0.463,1.000}
\definecolor{RoyalBlue2}{rgb}{0.263,0.431,0.933}
\definecolor{RoyalBlue3}{rgb}{0.227,0.373,0.804}
\definecolor{RoyalBlue4}{rgb}{0.153,0.251,0.545}
\definecolor{blue1}{rgb}{0.000,0.000,1.000}
\definecolor{blue2}{rgb}{0.000,0.000,0.933}
\definecolor{blue3}{rgb}{0.000,0.000,0.804}
\definecolor{blue4}{rgb}{0.000,0.000,0.545}
\definecolor{DodgerBlue1}{rgb}{0.118,0.565,1.000}
\definecolor{DodgerBlue2}{rgb}{0.110,0.525,0.933}
\definecolor{DodgerBlue3}{rgb}{0.094,0.455,0.804}
\definecolor{DodgerBlue4}{rgb}{0.063,0.306,0.545}
\definecolor{SteelBlue1}{rgb}{0.388,0.722,1.000}
\definecolor{SteelBlue2}{rgb}{0.361,0.675,0.933}
\definecolor{SteelBlue3}{rgb}{0.310,0.580,0.804}
\definecolor{SteelBlue4}{rgb}{0.212,0.392,0.545}
\definecolor{DeepSkyBlue1}{rgb}{0.000,0.749,1.000}
\definecolor{DeepSkyBlue2}{rgb}{0.000,0.698,0.933}
\definecolor{DeepSkyBlue3}{rgb}{0.000,0.604,0.804}
\definecolor{DeepSkyBlue4}{rgb}{0.000,0.408,0.545}
\definecolor{SkyBlue1}{rgb}{0.529,0.808,1.000}
\definecolor{SkyBlue2}{rgb}{0.494,0.753,0.933}
\definecolor{SkyBlue3}{rgb}{0.424,0.651,0.804}
\definecolor{SkyBlue4}{rgb}{0.290,0.439,0.545}
\definecolor{LightSkyBlue1}{rgb}{0.690,0.886,1.000}
\definecolor{LightSkyBlue2}{rgb}{0.643,0.827,0.933}
\definecolor{LightSkyBlue3}{rgb}{0.553,0.714,0.804}
\definecolor{LightSkyBlue4}{rgb}{0.376,0.482,0.545}
\definecolor{SlateGray1}{rgb}{0.776,0.886,1.000}
\definecolor{SlateGray2}{rgb}{0.725,0.827,0.933}
\definecolor{SlateGray3}{rgb}{0.624,0.714,0.804}
\definecolor{SlateGray4}{rgb}{0.424,0.482,0.545}
\definecolor{LightSteelBlue1}{rgb}{0.792,0.882,1.000}
\definecolor{LightSteelBlue2}{rgb}{0.737,0.824,0.933}
\definecolor{LightSteelBlue3}{rgb}{0.635,0.710,0.804}
\definecolor{LightSteelBlue4}{rgb}{0.431,0.482,0.545}
\definecolor{LightBlue1}{rgb}{0.749,0.937,1.000}
\definecolor{LightBlue2}{rgb}{0.698,0.875,0.933}
\definecolor{LightBlue3}{rgb}{0.604,0.753,0.804}
\definecolor{LightBlue4}{rgb}{0.408,0.514,0.545}
\definecolor{LightCyan1}{rgb}{0.878,1.000,1.000}
\definecolor{LightCyan2}{rgb}{0.820,0.933,0.933}
\definecolor{LightCyan3}{rgb}{0.706,0.804,0.804}
\definecolor{LightCyan4}{rgb}{0.478,0.545,0.545}
\definecolor{PaleTurquoise1}{rgb}{0.733,1.000,1.000}
\definecolor{PaleTurquoise2}{rgb}{0.682,0.933,0.933}
\definecolor{PaleTurquoise3}{rgb}{0.588,0.804,0.804}
\definecolor{PaleTurquoise4}{rgb}{0.400,0.545,0.545}
\definecolor{CadetBlue1}{rgb}{0.596,0.961,1.000}
\definecolor{CadetBlue2}{rgb}{0.557,0.898,0.933}
\definecolor{CadetBlue3}{rgb}{0.478,0.773,0.804}
\definecolor{CadetBlue4}{rgb}{0.325,0.525,0.545}
\definecolor{turquoise1}{rgb}{0.000,0.961,1.000}
\definecolor{turquoise2}{rgb}{0.000,0.898,0.933}
\definecolor{turquoise3}{rgb}{0.000,0.773,0.804}
\definecolor{turquoise4}{rgb}{0.000,0.525,0.545}
\definecolor{cyan1}{rgb}{0.000,1.000,1.000}
\definecolor{cyan2}{rgb}{0.000,0.933,0.933}
\definecolor{cyan3}{rgb}{0.000,0.804,0.804}
\definecolor{cyan4}{rgb}{0.000,0.545,0.545}
\definecolor{DarkSlateGray1}{rgb}{0.592,1.000,1.000}
\definecolor{DarkSlateGray2}{rgb}{0.553,0.933,0.933}
\definecolor{DarkSlateGray3}{rgb}{0.475,0.804,0.804}
\definecolor{DarkSlateGray4}{rgb}{0.322,0.545,0.545}
\definecolor{aquamarine1}{rgb}{0.498,1.000,0.831}
\definecolor{aquamarine2}{rgb}{0.463,0.933,0.776}
\definecolor{aquamarine3}{rgb}{0.400,0.804,0.667}
\definecolor{aquamarine4}{rgb}{0.271,0.545,0.455}
\definecolor{DarkSeaGreen1}{rgb}{0.757,1.000,0.757}
\definecolor{DarkSeaGreen2}{rgb}{0.706,0.933,0.706}
\definecolor{DarkSeaGreen3}{rgb}{0.608,0.804,0.608}
\definecolor{DarkSeaGreen4}{rgb}{0.412,0.545,0.412}
\definecolor{SeaGreen1}{rgb}{0.329,1.000,0.624}
\definecolor{SeaGreen2}{rgb}{0.306,0.933,0.580}
\definecolor{SeaGreen3}{rgb}{0.263,0.804,0.502}
\definecolor{SeaGreen4}{rgb}{0.180,0.545,0.341}
\definecolor{PaleGreen1}{rgb}{0.604,1.000,0.604}
\definecolor{PaleGreen2}{rgb}{0.565,0.933,0.565}
\definecolor{PaleGreen3}{rgb}{0.486,0.804,0.486}
\definecolor{PaleGreen4}{rgb}{0.329,0.545,0.329}
\definecolor{SpringGreen1}{rgb}{0.000,1.000,0.498}
\definecolor{SpringGreen2}{rgb}{0.000,0.933,0.463}
\definecolor{SpringGreen3}{rgb}{0.000,0.804,0.400}
\definecolor{SpringGreen4}{rgb}{0.000,0.545,0.271}
\definecolor{green1}{rgb}{0.000,1.000,0.000}
\definecolor{green2}{rgb}{0.000,0.933,0.000}
\definecolor{green3}{rgb}{0.000,0.804,0.000}
\definecolor{green4}{rgb}{0.000,0.545,0.000}
\definecolor{chartreuse1}{rgb}{0.498,1.000,0.000}
\definecolor{chartreuse2}{rgb}{0.463,0.933,0.000}
\definecolor{chartreuse3}{rgb}{0.400,0.804,0.000}
\definecolor{chartreuse4}{rgb}{0.271,0.545,0.000}
\definecolor{OliveDrab1}{rgb}{0.753,1.000,0.243}
\definecolor{OliveDrab2}{rgb}{0.702,0.933,0.227}
\definecolor{OliveDrab3}{rgb}{0.604,0.804,0.196}
\definecolor{OliveDrab4}{rgb}{0.412,0.545,0.133}
\definecolor{DarkOliveGreen1}{rgb}{0.792,1.000,0.439}
\definecolor{DarkOliveGreen2}{rgb}{0.737,0.933,0.408}
\definecolor{DarkOliveGreen3}{rgb}{0.635,0.804,0.353}
\definecolor{DarkOliveGreen4}{rgb}{0.431,0.545,0.239}
\definecolor{khaki1}{rgb}{1.000,0.965,0.561}
\definecolor{khaki2}{rgb}{0.933,0.902,0.522}
\definecolor{khaki3}{rgb}{0.804,0.776,0.451}
\definecolor{khaki4}{rgb}{0.545,0.525,0.306}
\definecolor{LightGoldenrod1}{rgb}{1.000,0.925,0.545}
\definecolor{LightGoldenrod2}{rgb}{0.933,0.863,0.510}
\definecolor{LightGoldenrod3}{rgb}{0.804,0.745,0.439}
\definecolor{LightGoldenrod4}{rgb}{0.545,0.506,0.298}
\definecolor{LightYellow1}{rgb}{1.000,1.000,0.878}
\definecolor{LightYellow2}{rgb}{0.933,0.933,0.820}
\definecolor{LightYellow3}{rgb}{0.804,0.804,0.706}
\definecolor{LightYellow4}{rgb}{0.545,0.545,0.478}
\definecolor{yellow1}{rgb}{1.000,1.000,0.000}
\definecolor{yellow2}{rgb}{0.933,0.933,0.000}
\definecolor{yellow3}{rgb}{0.804,0.804,0.000}
\definecolor{yellow4}{rgb}{0.545,0.545,0.000}
\definecolor{gold1}{rgb}{1.000,0.843,0.000}
\definecolor{gold2}{rgb}{0.933,0.788,0.000}
\definecolor{gold3}{rgb}{0.804,0.678,0.000}
\definecolor{gold4}{rgb}{0.545,0.459,0.000}
\definecolor{goldenrod1}{rgb}{1.000,0.757,0.145}
\definecolor{goldenrod2}{rgb}{0.933,0.706,0.133}
\definecolor{goldenrod3}{rgb}{0.804,0.608,0.114}
\definecolor{goldenrod4}{rgb}{0.545,0.412,0.078}
\definecolor{DarkGoldenrod1}{rgb}{1.000,0.725,0.059}
\definecolor{DarkGoldenrod2}{rgb}{0.933,0.678,0.055}
\definecolor{DarkGoldenrod3}{rgb}{0.804,0.584,0.047}
\definecolor{DarkGoldenrod4}{rgb}{0.545,0.396,0.031}
\definecolor{RosyBrown1}{rgb}{1.000,0.757,0.757}
\definecolor{RosyBrown2}{rgb}{0.933,0.706,0.706}
\definecolor{RosyBrown3}{rgb}{0.804,0.608,0.608}
\definecolor{RosyBrown4}{rgb}{0.545,0.412,0.412}
\definecolor{IndianRed1}{rgb}{1.000,0.416,0.416}
\definecolor{IndianRed2}{rgb}{0.933,0.388,0.388}
\definecolor{IndianRed3}{rgb}{0.804,0.333,0.333}
\definecolor{IndianRed4}{rgb}{0.545,0.227,0.227}
\definecolor{sienna1}{rgb}{1.000,0.510,0.278}
\definecolor{sienna2}{rgb}{0.933,0.475,0.259}
\definecolor{sienna3}{rgb}{0.804,0.408,0.224}
\definecolor{sienna4}{rgb}{0.545,0.278,0.149}
\definecolor{burlywood1}{rgb}{1.000,0.827,0.608}
\definecolor{burlywood2}{rgb}{0.933,0.773,0.569}
\definecolor{burlywood3}{rgb}{0.804,0.667,0.490}
\definecolor{burlywood4}{rgb}{0.545,0.451,0.333}
\definecolor{wheat1}{rgb}{1.000,0.906,0.729}
\definecolor{wheat2}{rgb}{0.933,0.847,0.682}
\definecolor{wheat3}{rgb}{0.804,0.729,0.588}
\definecolor{wheat4}{rgb}{0.545,0.494,0.400}
\definecolor{tan1}{rgb}{1.000,0.647,0.310}
\definecolor{tan2}{rgb}{0.933,0.604,0.286}
\definecolor{tan3}{rgb}{0.804,0.522,0.247}
\definecolor{tan4}{rgb}{0.545,0.353,0.169}
\definecolor{chocolate1}{rgb}{1.000,0.498,0.141}
\definecolor{chocolate2}{rgb}{0.933,0.463,0.129}
\definecolor{chocolate3}{rgb}{0.804,0.400,0.114}
\definecolor{chocolate4}{rgb}{0.545,0.271,0.075}
\definecolor{firebrick1}{rgb}{1.000,0.188,0.188}
\definecolor{firebrick2}{rgb}{0.933,0.173,0.173}
\definecolor{firebrick3}{rgb}{0.804,0.149,0.149}
\definecolor{firebrick4}{rgb}{0.545,0.102,0.102}
\definecolor{brown1}{rgb}{1.000,0.251,0.251}
\definecolor{brown2}{rgb}{0.933,0.231,0.231}
\definecolor{brown3}{rgb}{0.804,0.200,0.200}
\definecolor{brown4}{rgb}{0.545,0.137,0.137}
\definecolor{salmon1}{rgb}{1.000,0.549,0.412}
\definecolor{salmon2}{rgb}{0.933,0.510,0.384}
\definecolor{salmon3}{rgb}{0.804,0.439,0.329}
\definecolor{salmon4}{rgb}{0.545,0.298,0.224}
\definecolor{LightSalmon1}{rgb}{1.000,0.627,0.478}
\definecolor{LightSalmon2}{rgb}{0.933,0.584,0.447}
\definecolor{LightSalmon3}{rgb}{0.804,0.506,0.384}
\definecolor{LightSalmon4}{rgb}{0.545,0.341,0.259}
\definecolor{orange1}{rgb}{1.000,0.647,0.000}
\definecolor{orange2}{rgb}{0.933,0.604,0.000}
\definecolor{orange3}{rgb}{0.804,0.522,0.000}
\definecolor{orange4}{rgb}{0.545,0.353,0.000}
\definecolor{DarkOrange1}{rgb}{1.000,0.498,0.000}
\definecolor{DarkOrange2}{rgb}{0.933,0.463,0.000}
\definecolor{DarkOrange3}{rgb}{0.804,0.400,0.000}
\definecolor{DarkOrange4}{rgb}{0.545,0.271,0.000}
\definecolor{coral1}{rgb}{1.000,0.447,0.337}
\definecolor{coral2}{rgb}{0.933,0.416,0.314}
\definecolor{coral3}{rgb}{0.804,0.357,0.271}
\definecolor{coral4}{rgb}{0.545,0.243,0.184}
\definecolor{tomato1}{rgb}{1.000,0.388,0.278}
\definecolor{tomato2}{rgb}{0.933,0.361,0.259}
\definecolor{tomato3}{rgb}{0.804,0.310,0.224}
\definecolor{tomato4}{rgb}{0.545,0.212,0.149}
\definecolor{OrangeRed1}{rgb}{1.000,0.271,0.000}
\definecolor{OrangeRed2}{rgb}{0.933,0.251,0.000}
\definecolor{OrangeRed3}{rgb}{0.804,0.216,0.000}
\definecolor{OrangeRed4}{rgb}{0.545,0.145,0.000}
\definecolor{red1}{rgb}{1.000,0.000,0.000}
\definecolor{red2}{rgb}{0.933,0.000,0.000}
\definecolor{red3}{rgb}{0.804,0.000,0.000}
\definecolor{red4}{rgb}{0.545,0.000,0.000}
\definecolor{DeepPink1}{rgb}{1.000,0.078,0.576}
\definecolor{DeepPink2}{rgb}{0.933,0.071,0.537}
\definecolor{DeepPink3}{rgb}{0.804,0.063,0.463}
\definecolor{DeepPink4}{rgb}{0.545,0.039,0.314}
\definecolor{HotPink1}{rgb}{1.000,0.431,0.706}
\definecolor{HotPink2}{rgb}{0.933,0.416,0.655}
\definecolor{HotPink3}{rgb}{0.804,0.376,0.565}
\definecolor{HotPink4}{rgb}{0.545,0.227,0.384}
\definecolor{pink1}{rgb}{1.000,0.710,0.773}
\definecolor{pink2}{rgb}{0.933,0.663,0.722}
\definecolor{pink3}{rgb}{0.804,0.569,0.620}
\definecolor{pink4}{rgb}{0.545,0.388,0.424}
\definecolor{LightPink1}{rgb}{1.000,0.682,0.725}
\definecolor{LightPink2}{rgb}{0.933,0.635,0.678}
\definecolor{LightPink3}{rgb}{0.804,0.549,0.584}
\definecolor{LightPink4}{rgb}{0.545,0.373,0.396}
\definecolor{PaleVioletRed1}{rgb}{1.000,0.510,0.671}
\definecolor{PaleVioletRed2}{rgb}{0.933,0.475,0.624}
\definecolor{PaleVioletRed3}{rgb}{0.804,0.408,0.537}
\definecolor{PaleVioletRed4}{rgb}{0.545,0.278,0.365}
\definecolor{maroon1}{rgb}{1.000,0.204,0.702}
\definecolor{maroon2}{rgb}{0.933,0.188,0.655}
\definecolor{maroon3}{rgb}{0.804,0.161,0.565}
\definecolor{maroon4}{rgb}{0.545,0.110,0.384}
\definecolor{VioletRed1}{rgb}{1.000,0.243,0.588}
\definecolor{VioletRed2}{rgb}{0.933,0.227,0.549}
\definecolor{VioletRed3}{rgb}{0.804,0.196,0.471}
\definecolor{VioletRed4}{rgb}{0.545,0.133,0.322}
\definecolor{magenta1}{rgb}{1.000,0.000,1.000}
\definecolor{magenta2}{rgb}{0.933,0.000,0.933}
\definecolor{magenta3}{rgb}{0.804,0.000,0.804}
\definecolor{magenta4}{rgb}{0.545,0.000,0.545}
\definecolor{orchid1}{rgb}{1.000,0.514,0.980}
\definecolor{orchid2}{rgb}{0.933,0.478,0.914}
\definecolor{orchid3}{rgb}{0.804,0.412,0.788}
\definecolor{orchid4}{rgb}{0.545,0.278,0.537}
\definecolor{plum1}{rgb}{1.000,0.733,1.000}
\definecolor{plum2}{rgb}{0.933,0.682,0.933}
\definecolor{plum3}{rgb}{0.804,0.588,0.804}
\definecolor{plum4}{rgb}{0.545,0.400,0.545}
\definecolor{MediumOrchid1}{rgb}{0.878,0.400,1.000}
\definecolor{MediumOrchid2}{rgb}{0.820,0.373,0.933}
\definecolor{MediumOrchid3}{rgb}{0.706,0.322,0.804}
\definecolor{MediumOrchid4}{rgb}{0.478,0.216,0.545}
\definecolor{DarkOrchid1}{rgb}{0.749,0.243,1.000}
\definecolor{DarkOrchid2}{rgb}{0.698,0.227,0.933}
\definecolor{DarkOrchid3}{rgb}{0.604,0.196,0.804}
\definecolor{DarkOrchid4}{rgb}{0.408,0.133,0.545}
\definecolor{purple1}{rgb}{0.608,0.188,1.000}
\definecolor{purple2}{rgb}{0.569,0.173,0.933}
\definecolor{purple3}{rgb}{0.490,0.149,0.804}
\definecolor{purple4}{rgb}{0.333,0.102,0.545}
\definecolor{MediumPurple1}{rgb}{0.671,0.510,1.000}
\definecolor{MediumPurple2}{rgb}{0.624,0.475,0.933}
\definecolor{MediumPurple3}{rgb}{0.537,0.408,0.804}
\definecolor{MediumPurple4}{rgb}{0.365,0.278,0.545}
\definecolor{thistle1}{rgb}{1.000,0.882,1.000}
\definecolor{thistle2}{rgb}{0.933,0.824,0.933}
\definecolor{thistle3}{rgb}{0.804,0.710,0.804}
\definecolor{thistle4}{rgb}{0.545,0.482,0.545}
\definecolor{gray5}{rgb}{0.051,0.051,0.051}
\definecolor{gray10}{rgb}{0.102,0.102,0.102}
\definecolor{gray15}{rgb}{0.149,0.149,0.149}
\definecolor{gray20}{rgb}{0.200,0.200,0.200}
\definecolor{gray25}{rgb}{0.251,0.251,0.251}
\definecolor{gray30}{rgb}{0.302,0.302,0.302}
\definecolor{gray35}{rgb}{0.349,0.349,0.349}
\definecolor{gray40}{rgb}{0.400,0.400,0.400}
\definecolor{gray45}{rgb}{0.451,0.451,0.451}
\definecolor{gray50}{rgb}{0.498,0.498,0.498}
\definecolor{gray55}{rgb}{0.549,0.549,0.549}
\definecolor{gray60}{rgb}{0.600,0.600,0.600}
\definecolor{gray65}{rgb}{0.651,0.651,0.651}
\definecolor{gray70}{rgb}{0.702,0.702,0.702}
\definecolor{gray75}{rgb}{0.749,0.749,0.749}
\definecolor{gray80}{rgb}{0.800,0.800,0.800}
\definecolor{gray85}{rgb}{0.851,0.851,0.851}
\definecolor{gray90}{rgb}{0.898,0.898,0.898}
\definecolor{gray95}{rgb}{0.949,0.949,0.949}
\definecolor{gray100}{rgb}{1.000,1.000,1.000}
\definecolor{DarkGray}{rgb}{0.663,0.663,0.663}
\definecolor{DarkBlue}{rgb}{0.000,0.000,0.545}
\definecolor{DarkCyan}{rgb}{0.000,0.545,0.545}
\definecolor{DarkMagenta}{rgb}{0.545,0.000,0.545}
\definecolor{DarkRed}{rgb}{0.545,0.000,0.000}
\definecolor{LightGreen}{rgb}{0.565,0.933,0.565}

\psset{unit=1mm}
\usepackage{subfig}

\usepackage{accents}
\newcommand{\ubar}[1]{\underaccent{\bar}{#1}}

\newcommand{\highlight}[1]{#1} 
\renewcommand{\vec}[1]{\mathbf{#1}}

\newcommand{\highlighta}[1]{#1}
\begin{document}
\begin{titlepage}

\title{\bf Selling Wind}

\author{Ali Kakhbod\\ MIT \and Asuman 
Ozdaglar\\ MIT \and  Ian Schneider\thanks{Kakhbod: Department of Economics, Massachusetts Institute of Technology (MIT),
Cambridge, MA 02139, USA. Email: \url{akakhbod@mit.edu}.
Ozdaglar: Electrical Engineering and Computer Science (EECS), Laboratory for Information and Decision Systems (LIDS), Institute for Data, Systems, and Society (IDSS), Massachusetts
Institute of Technology (MIT), Cambridge, MA 02139, USA. Email: %
\url{asuman@mit.edu}.
Schneider: Institute for Data, Systems, and Society (IDSS), Massachusetts Institute of Technology (MIT), Cambridge, MA 02139, USA. Email: \url{ischneid@mit.edu }. We are grateful to Daron Acemoglu for valuable comments and suggestions. We also thank Matt Butner for sharing his data on wind availability in MISO and participants at various seminars and conferences for useful feedback and suggestions.
 }\\ MIT}

\date{December 2018}

\maketitle
\thispagestyle{empty}

\begin{abstract}
We offer a parsimonious model to investigate how strategic wind producers sell energy under stochastic production constraints, where the extent of heterogeneity of wind energy availability varies according to wind farm locations.
The main insight of our analysis is that increasing heterogeneity in resource availability improves social welfare, as a function of its effects both on improving diversification and on reducing withholding by firms. We show that this insight is quite robust for any concave and downward-sloping inverse demand function. 
The model is also used to analyze the effect of heterogeneity on firm profits and opportunities for collusion. Finally, we analyze the impacts of improving public information and weather forecasting; enhanced public forecasting increases welfare, but it is not always in the best interests of strategic producers.

\bigskip

\bigskip


\end{abstract}



\end{titlepage}

\newpage
\pagestyle{plain}
\setcounter{page}{2}



\section{Introduction} \label{Sec::Introduction}

The market share and total production of renewable electricity is growing rapidly. In March 2017, wind energy was responsible for 8.6\% of U.S. electricity generation, doubling its market share and total production from five years prior. Renewable electricity is a critical component of global efforts to reduce carbon dioxide emissions, and its cost is rapidly declining.

Prominent sources of renewable electricity - wind and solar energy - have stochastic resource availability: it is not possible to perfectly predict when there will be wind or sunshine available for energy production. The spatial and temporal variability of renewable energy resources has a significant impact on their value to society \citep{Joskow2006, hirth2013market, hirth2016wind}. 
Furthermore, since wind production reduces local prices due to the merit order effect, highly correlated local wind energy availability reduces the average value of wind energy produced \citep{woo2011impact, Ketterer2014}.

\highlighta{Existing literature focuses on strategic behavior in electricity markets without substantial amounts of renewable energy. Research on market power in the electricity sector \citep{joskow1988markets} provided important insight for electricity system deregulation. Electricity system market power research does not traditionally focus on stochasticity because fossil-fuel generators do not have significant resource uncertainty. Instead, it focuses on other key features of the electricity sector that impact market power, like transmission constraints \citep{cardell1997market}, financial transmission rights \citep{joskow2000transmission}, and market price caps \citep{Joskow2007}. \cite{Acemoglu2017} establish that diverse ownership portfolios of renewable and thermal generation by strategic firms may be welfare reducing, because they can reduce (or even neutralize) the merit order effect. \cite{butner2018gone} provides empirical evidence of these effects.}  
	
	\highlighta{We are interested in how a particular characteristic of renewable energy resources--the stochastic dependence of resource availability across firms--impacts strategic behavior, market power, and welfare.} The link between stochastic heterogeneity\footnote{Stochastic heterogeneity refers to the level of stochastic dependence. Throughout the main body of the paper, we use a term ``dispersion'' to succinctly refer to the extent of stochastic dependence. In the linear case, high (low) stochastic dependence is equivalent to highly correlated (uncorrelated) stochastic resource availability across different producers.} of resource availability and welfare is an important area for research because various policies impact the investment strategies of wind producers and therefore the stochastic characteristics of the wind energy portfolio in a given region \citep{Kok2016, schneider2017wind}. Common subsidy forms for renewable energy, like the production tax credit (PTC) and state-level renewable portfolio standards (RPS), impact renewable energy investments \citep{fischer2010renewable}. 
	Figure \ref{wind_compare} shows probability distributions for two different wind farms in the MISO region, conditional on the output of a third wind farm $i$ in the MISO region. The nature of stochastic dependence is very different for each pair of wind farms. Wind farm $i$'s output is highly correlated to the output of the wind farm displayed on the right, and essentially uncorrelated with the output of the wind farm displayed on the left.\footnote{For the purposes of this graph, we use measured production as a proxy for resource availability. Here and throughout the paper we stylize resource availability as discrete, i.e. the resource availability at $i$ $w_i = L$ or $w_i = H$. Since real world availability is continuous, for this figure, we say that $w_i = L$ when the wind availability is less than 3\% of its maximum availability (bottom 27\% of periods) and $w_i = H$ when the wind availability is greater than 67\% of its maximum availability (top 20\% of periods).}

Clearly, policy changes can impact investment strategies for renewable energy and the characteristics of system-wide resource uncertainty. This begs the questions: \emph{Is it important to encourage policies that increase the heterogeneity of stochastic resources? Will investment in wind energy naturally lead to the level of resource heterogeneity that maximizes social welfare?} Just as policy makers seek to limit market concentration in certain industries, they might support policies to increase stochastic heterogeneity of renewable resources in the electric power industry.
These efforts have growing import because existing strategies for market power monitoring in electricity markets will be challenged by an influx of renewable generation, since regulators have imperfect information regarding resource availability and risk preferences for firms that own stochastic generation.\footnote{Additionally, \citet{munoz2018economic} discuss challenges associated with auditing the opportunity costs of traditional generators in markets with physical inflexibilities and non-convex costs. These issues have growing import in systems with high levels of renewable energy.}

\begin{figure}[h]
	\centering
	\includegraphics[scale = 0.48]{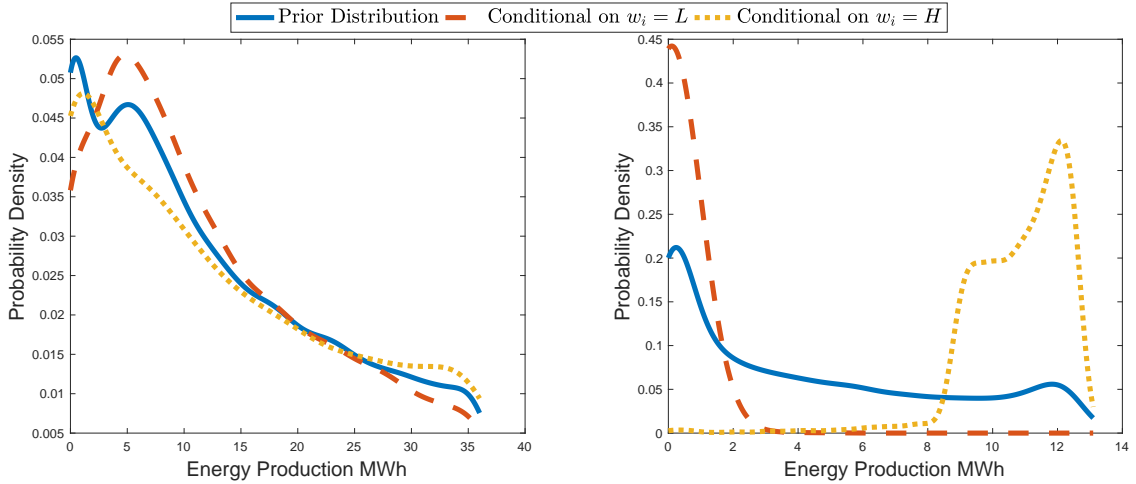}
	\caption{Prior and conditional empirical distributions for resource availability from two wind farms in MISO, based on hourly energy production from 2014 through 2016. The conditional distributions are conditioned on the resource availability at a third wind farm $i$ in MISO. The right plot shows an example of highly correlated resource availability from one pair of wind farms, and the left plot shows an example of uncorrelated resource availability from another pair of wind farms.}
	\label{wind_compare}
\end{figure}

\highlighta{We study strategic firms with private information regarding their realized energy availability, or ``state.'' This energy availability is equivalent to a production constraint, because it limits the extent of production by the firm in any given period.  Since the resource availability of wind energy is uncertain, from an individual firm's perspective its competitors' production constraints are stochastic. However, the resource availability of wind energy has a high degree of stochastic dependence; firms can gain important information about their competitors' production constraints from the realization of their own resource availability. As such, the extent of stochastic dependence regarding firms' resource availability becomes an important factor that impacts strategic behavior, market power, and welfare.} For clarity, we focus on wind energy, but the insights can be extended directly to solar energy or any other resource with stochastic availability and negligible marginal costs. 



\highlighta{We model producer competition as an incomplete information Cournot game with correlated types, where the type refers to the stochastic resource availability (production constraint) that is private information for each individual producer.} The base model uses a Cournot duopoly market. We utilize a parameter $d$ to represent the level of heterogeneity amongst wind producers; throughout, we refer to $d$ as the level of dispersion.  This notion is similar to geographic distance; research has shown that correlation in wind availability across pairs of wind producers is generally decreasing in geographic distance \citep{Sinden2007}. 

The results provide clear insight to explain how stochastic resource heterogeneity can impact welfare in imperfect electricity markets. Increasing heterogeneity in wind resource availability is beneficial for two distinct reasons: it increases the diversification of resources, but it also reduces strategic withholding because it changes the information that a producer's own energy availability provides about the likely energy availability of the other firms in the market. The results of our model imply that imperfect competition in energy markets can affect investment in renewable energy, resulting in a system with sub-optimal levels of resource heterogeneity.\footnote{Other features of electricity markets might also reduce heterogeneity of wind resources or distance between wind producers, including the presence of existing transmission, variance in state-level renewable policies, and quantity-based subsidies. However, we focus specifically on market failures due to imperfect competition.}    

\highlighta{We extend the benchmark model to an arbitrary number of producers and examine the social welfare of the resulting Cournot-Nash equilibria. We define general probability distributions on the players' types (their stochastic production constraints). We establish sufficient conditions for a pair of probability distributions such that one distribution will dominate another in terms of the social welfare of their associated Cournot-Nash equilibria, for any inverse demand function that is decreasing and concave. This defines a partial ordering of joint distributions in terms of the social welfare of their resulting equilibria. The model is also extended to examine the impacts of traditional generation. In an imperfect market with traditional generation, welfare is again increasing in resource heterogeneity.}  

Next, the model is utilized to examine the effects of heterogeneity on collusion and on policies to prevent collusion. 
\highlighta{If they do not face potential penalties for collusion, firms with stochastic availability will always choose to collude because they benefit from sharing information and from sharing monopolistic profits. Increasing heterogeneity of wind production has a range of impacts on collusion, impacting its value to producers, the costs of collusion on social welfare, and the level of enforcement required to prevent collusion.} 

Finally, we investigate the effects of public sharing of high-quality weather forecasts, using the limiting case whereby the true realized energy availability of firms is monitored and shared. The results suggest an important contradiction: information sharing through improved forecasting is socially beneficial, but it does not always improve producer profits. As such, it will not necessarily be undertaken by producers acting in their own best interest. Since firms have stochastic availability, high quality weather forecasting can be undertaken publicly in order to maximize social welfare.  

\highlighta{These results provide a framework for evaluating policies that impact investment and information-provision in imperfectly competitive markets, like electricity markets. The results can help us understand how policies that impact the dispersion of renewable energy resources, and thus the characteristics of stochastic energy availability, ultimately impact welfare in imperfectly competitive electricity markets.}
 
\subsection{Literature Review}

\highlighta{The focus of literature has been} on the impacts of wind diversification or total wind variance on average electricity prices or on the cost of wind integration. Increased heterogeneity of wind resources has at least three impacts on social welfare and the cost of electricity: 
\begin{itemize}
	\item[1. ] Balancing costs for managing wind production. \emph{This impact is well-studied in the literature and not covered in the model in this paper.}
	\item[2. ] The average benefit of wind production. \emph{This is discussed in some existing literature, but we provide a new formal model that provides insights on its impact on welfare.}
	\item[3. ] Strategic curtailment by wind and traditional power producers.  \emph{This impact has not been proposed in previous literature. This paper formalizes the concept and explains its impact on welfare.}
\end{itemize} 

First, \highlighta{increased wind} heterogeneity decreases balancing costs because it reduces hour-to-hour fluctuations in total wind energy production and because it reduces net uncertainty of availability in a given hour. This impact is well-studied in existing literature. \citet{Fertig2012} show that increasing diversification reduces the average hourly fluctuation in total power output and increases the equivalent firm power production. 

These short-term time-dependent impacts are not covered specifically in our model, which ignores the complexities associated with sequential market clearings\footnote{Prices in sequential markets are also impacted by market power, which helps explain why prices in sequential markets sometimes diverge \citep{ito2016sequential}.}. However, the benefits of increased heterogeneity likely have a net positive impact on welfare due to a reduction of balancing costs, so they do not change the general direction of the main welfare results. 

Second, increasing heterogeneity increases the average benefit of wind production. Increasing levels of wind generation have been shown to reduce prices in Germany and in West Texas \citep{Ketterer2014, woo2011impact}. In general, wind has declining marginal benefits because of the convexity of the electricity supply curve from traditional generators, which serve as strategic substitutes for wind energy. We model this effect by assuming that utility obtained from consuming wind energy is concave. Our model mirrors the basic empirical result; when more wind is produced, in a single period or on average, lower prices result. While existing research focuses on the price impacts of additional wind penetration, the stylized model described herein allows us to extend the results by focusing specifically on the impacts of heterogeneity on social welfare.   
 
Finally, the third impact of higher wind heterogeneity is its effect on the ability of wind generators to strategically curtail their energy production, when these generators have market power but some uncertainty regarding their competitors' production constraints. To our knowledge, this paper represents the first time that this effect of resource heterogeneity has been proposed and analyzed.

This paper uses a Cournot model to analyze the affects of dispersion on bidding behavior and welfare in a market with stochastic energy availability and private information. The Cournot assumption provides a simple model of imperfect competition, which is an important feature of markets with renewable generation: as firms operate more wind and solar generation, it will become increasingly difficult to prevent the exercise of market power, due to the uncertainty in underlying resource availability. The Cournot model is a useful simplification. In practice, firms submit a supply function that specifies how much energy they are willing to offer at any given price. \footnote{The Cournot setup is considered a good approximation to real-world electricity markets \citep{Hogan1997, Oren1997, BorensteinBushnellKnittel.2002, Willems.Rumiantseva.Weigt2009}. There are other ways to model producer offers in electricity markets, including supply function offers \citep{Anderson2002, holmberg2007supply}. \citet{wolfram1998strategic} and \citet{hortacsu2008understanding} offer empirical analyses of strategic bidding in multi-unit electricity auctions. \citet{Willems.Rumiantseva.Weigt2009}  and \citet{Ventosa.Baill.Rivier2005} discuss the comparative benefits of Cournot models versus the full supply function model.} 

\highlighta{Our model for strategic firms is a Cournot-Nash game with incomplete information. There is a substantial economics literature on this subject. \cite{einy2010existence} survey the literature and explain the conditions for equilibria existence and uniqueness in Cournot-Nash games with incomplete information and correlated types. Nearly all of this literature treat firms' objective functions as stochastic; we focus on the case where production constraints are stochastic. \cite{richter2013incomplete} also study the topic of Cournot games with stochastic production constraints and incomplete information, but they focus on firms that are stochastically independent. Stochastic dependence of the firms' production constraints has major impacts on the results, including the value of information sharing. Our research does not focus on conditions for existence; reasonable assumptions for the electricity sector generally lead to the existence of equilibria. Instead, we focus on developing new results to link the extent of stochastic dependence to strategic behavior and welfare in the equilibria.}

The main idea of this research is to formalize game-theoretic equilibria where producers have stochastic and dependent production constraints, in order to examine the effects of correlation in resource availability on the resulting equilibria. We consider the case of multiple wind producers offering their energy into markets, when their maximum availability is stochastic and correlated amongst producers. Existing research studies energy market equilibria in other ways. For instance, \citet{Hobbs2007} examine the effects of joint constraints and non-smooth demand functions. \citet{Downward2010} and \citet{Yao.Adler.Oren2008} study Cournot equilibria in markets with transmission constraints. \citet{de2016comparison} study the effects of Cournot competition on the efficacy and impacts of various renewable energy incentive methods. \citet{gilotte2006investments} and \citet{pineau2011dynamic} study investment in energy markets with Cournot competition. Compared to the aforementioned literature, we abstract many important notions of real-world electricity systems in order to clearly focus on our question of interest, which is not addressed in the existing literature: \emph{how does heterogeneity in stochastic renewable energy availability impact market power and social welfare?} 

Section \ref{Sec::Model} introduces the benchmark duopoly model, and Section \ref{sec :: equil} describes the features of the Cournot equilibrium. Section \ref{Sec::impacts} describes the impacts of wind heterogeneity on the diversification of wind energy and on strategic curtailment by wind producers. These impacts drive many of the main results presented herein. Sections \ref{Sec::welfare} and \ref{Sec::price} describe the effects of heterogeneity on social welfare, price, and profits in the duopoly model. Sections \ref{Sec::multw} and \ref{sec :: trad} extend the results to the case of an oligopoly market with multiple wind generators and with both wind and traditional generators, respectively. Section \ref{Sec:collusion} examines the effect of heterogeneity on collusion and on the cost of efforts to prevent collusion. Section \ref{sec :: infosharing} describes how the level of heterogeneity impacts the likelihood that firms will choose to publicly share information, and shows that public information sharing is always socially beneficial. Section \ref{Sec::conclusion} concludes.  

\section{Benchmark Model} \label{Sec::Model}
%

Consider two wind energy producers in a market with imperfect competition, operating two locally separate wind farms to generate energy. For each producer $i$, the maximum available  wind energy, $w_i$, is \textbf{stochastic} and might be either $H$ (high) or $L$ (low), with $H>L$ and with probability $\Pr\{w_i=H\}=\beta=1-\Pr\{w_i=L\} > 0, \ i\in\{1,2\}$. When $w_i = H$ $(w_i = L)$, we say that producer $i$ is in the high (low) state. Let $d\in[0, 1]$ be the \textbf{dispersion}  between the two wind producers, where the maximum dispersion is normalized to $1$. The parameter $d$  captures the extent of {\bf heterogeneity} in terms of wind energy availability for these wind producers. When $d$ is small, the availability of wind energy is highly correlated amongst wind producers. When one wind producer in the high state, the other wind producer is likely to be in the high state as well (similarly for the low state). However, when $d$ is high these locations become highly heterogeneous in terms of wind availability, so that extent of wind energy available to one producer does not reveal much information about the other wind producer's availability. In the case of high heterogeneity, the extent of wind energy availability is nearly or (in the limit) fully independent across wind producers. 

\highlighta{This section models the joint probability distribution of the available wind energy in a simple parameterized form.}  Precisely, for $i,j\in \{1,2\}$, the conditional probability of high wind availability is given by \eqref{prob}.
\begin{equation} \label{prob}
	\begin{aligned}
		\Pr\{w_i=H|w_j=H\}&=\frac{\beta}{\beta+{\color{Maroon}{d}}(1-\beta)}  \\
		\Pr\{w_i=H|w_j=L\}&=\frac{{\color{Maroon}{d}} \beta}{\beta+{\color{Maroon}{d}}(1-\beta)}
	\end{aligned}
\end{equation} 

When the wind producers are ``far'' from each other, $d = 1$,  we are in the limiting case of independent production; $\Pr\{w_i=H|w_j=H\} = \Pr\{w_i=H\} = \beta$ and $\Pr\{w_i=H|w_j=L\} = \beta$. On the other hand, when they are locally ``close'', $d = 0$,  we are in the full information case and $\Pr\{w_i=H|w_j=H\} = 1$. \highlighta{Section \ref{Sec::multw} extends the results in this section for arbitrary joint probability distributions for wind availability from multiple producers.}\footnote{\highlighta{In this context, we can think of $\beta$ as a forecast of the energy availability for each firm, using public information or information with negligible cost. When a firm realizes its own (private) energy availability, this new information changes its forecast of the energy available to its competitors, as shown in \eqref{prob}.}} 

 \begin{figure}[h]
 \begin{center}
\includegraphics[width=1\textwidth]{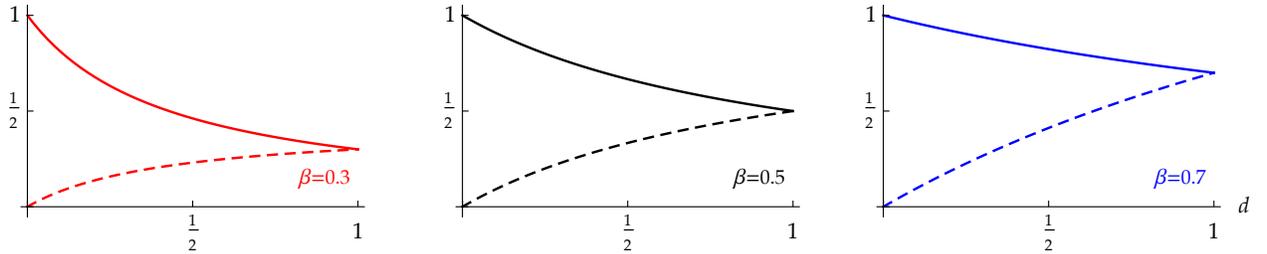}
\end{center}
\caption{The conditional probability distributions from \eqref{prob} for $d \in (0,1)$, for different values of the prior $\beta$. For each graph, the solid line represents $\Pr\{H|H\}$ and the dashed line represents $\Pr\{H|L\}$.}
\end{figure} \label{fig :: posterior}

Note that the extent of the difference between the high and low states corresponds to the extent of the variance in wind availability for each individual producer. For instance, if we fix the value of $H$ (e.g. as the capacity value of the wind producer), then variance of wind energy availability $\text{Var}(w_i)$ is monotonically decreasing in $L$.  

Let $q_i$ denote the amount of wind energy generated by producer $i\in\{1,2\}$. \highlighta{We assume the inverse demand $P:\mathbb{R}\to \mathbb{R}$ as a function of total supply $Q = q_1 + q_2$ is concave and downward, i.e., $P'(Q)<0, P''(Q) \leq 0$ for all $Q$. The marginal cost of production via wind energy is negligible.} Our model simplifies the electricity market model by focusing on a single real-time market with inverse demand $P(Q)$. We ignore the impacts of short- and long-term forward markets, e.g. day-ahead markets and capacity markets. While these markets are undoubtedly important, we focus on the real-time spot market because planned real-time bidding behavior ultimately informs strategy in forward markets.\footnote{While we refer to $w_i$ as the wind energy availability for firm $i$, readers can also think of $w_i$ as the realized error, the difference between energy availability in real-time and the day-ahead offer or prior forecast.}

The  producers compete \highlighta{in Cournot fashion}. According to its \textbf{private information} about its maximum available wind, $w_i\in \{H,L\}$, producer $i$ chooses $q_i(w_i)$ \highlighta{maximizing the expected value of its profit $\pi_i$, conditional on its realization of $w_i$:}
\begin{align*}
\E_{w_j}[\pi_i|w_i]&=\E\left[q_i \ P(q_i(w_i)+q_j(w_j)) \ | w_i\right], \\
 \quad \text{s.t.} &   \quad q_i(w_i)\in [0, w_i] 
\end{align*}

\section{Equilibrium} \label{sec :: equil}
To ensure that wind producers produce at full capacity in the low state (i.e. no curtailment when $w_i = L$), \highlighta{and to avoid equilibria where wind producers produce at full capacity in both states} we adopt the following assumption:
\begin{Assumption}\label{assumption-11}
Let $P(\cdot)$ be the inverse demand. Then $P(2L)+LP'(2L)>0$ and $P(H)+HP'(H)<0$.
\end{Assumption}

\highlighta{This assumption allows us to focus on equilibria where producers exercise strategic withholding in one state but not in the other. This represents the case of interest where the stochastic nature of the wind resource has important impacts on the equilibrium strategy. Under the above assumption the equilibrium is characterized as follows. }

\begin{Proposition}\label{pro0} Let $P'<0, P''\leq0$. Then, 
there exists a unique symmetric Bayesian Nash Equilibrium (BNE) such that 
\begin{align*}
q_i(w_i)=q(w_i)=\min\{w_i,\phi\} \quad \quad w_i\in \{L,H\}, \ i=1,2
\end{align*}
 where $\phi>L$ is the unique root of the following equation
 \begin{align*}
 \Pr\{L|H\}\big[P(L+\phi)+\phi P'(\phi+L)\big]+\Pr\{H|H\}\big[P(2\phi)+\phi P'(2\phi)\big]=0.
 \end{align*} 

\end{Proposition}
The proposition establishes the symmetric Bayesian Nash Equilibrium (BNE) for the benchmark model. In this equilibrium, firms produce $L$ in the low state and $\phi < H$ when they are in the high state. The proof for this proposition is provided in the Appendix; throughout the paper, all omitted proofs are in the Appendix. The intuition is that in the symmetric equilibria producers curtail based on the expected value of the first order condition, given the uncertain state of their competitor and their competitor's equivalent curtailment strategy.    

\noindent{{\color{Maroon}{\bf Example (Linear inverse demand):}}} \quad
To clarify understanding regarding Assumption \ref{assumption-11}, consider the case of linear inverse demand, i.e. $P(Q)=s-q_1-q_2$, where $Q$ denotes the sum of the firms' production $Q = q_1 + q_2$. Suppose there is no capacity constraint; then there exists a unique symmetric equilibrium $q_C$ (the Cournot level) in which the optimal production is given by
\begin{equation*}
q_1=q_2=q_C\equiv\frac{s}{3}<q_M\equiv \frac{s}{2},
\end{equation*}
where $q_M$ is the corresponding  Monopoly level.\footnote{Note that $q_C$ is the optimal strategy when $\pi_i=q_i(s-q_1-q_2)$ and $q_M$ is the optimal monopoly quantity maximizing $\pi=q(s-q)$.} 
Thus, with linear inverse demand, Assumption \ref{assumption-11} simply says that $L$ is lower than the Cournot level and $H$ is higher than the monopoly level, i.e. 
\[L<q_C<q_M<H.
\]
\highlighta{If the first part of the assumption is violated, wind producers always produce at the Cournot level $q_C$; the stochastic nature of the wind resource has no impact on the equilibrium strategy. If the second part of the assumption is violated, then wind producers would curtail in any situation, even absent a competitor.} Moreover, with linear inverse demand, the equilibrium can be explicitly characterized as follows. 
\begin{Corollary}\label{pro1}
Let the inverse demand be linear, i.e. $P(q_1+q_2)=s-q_1-q_2$. Then, there exists a unique symmetric pure-strategy Bayes Nash equilibrium such that
\begin{align*}
q_i(w_i)=q(w_i)=\min\left\{w_i, \phi\right\},  \quad w_i\in\{L,H\}, \  i=1,2, 
\end{align*} 
where $\phi \equiv \frac{s\beta+(s-L)(1-\beta)d}{3\beta+2(1-\beta)d}$.
\end{Corollary}
\noindent The subsequent sections introduce key effects that drive the impacts of $d$ on the equilibrium and its resulting impacts on welfare, price, and profits. 
\section{\highlight{Strategic Curtailment and Diversification}} \label{Sec::impacts}

This section explains useful Lemmas to help illustrate the two major impacts of dispersion $d$ in the strategic setting. Recall that $\phi=q(H)$ is the production when a firm is in the high state. When a firm is in the low state it produces $L$. Unless otherwise specified, all the following results hold for concave and downward inverse demand functions (i.e. $P'<0, P''\leq 0$) satisfying Assumption \ref{assumption-11}.
\begin{Lemma}\label{lem1}
As $d$ increases production in the high state increases, i.e. $\frac{\p \phi}{\p d}>0$.
\end{Lemma}

The intuition derives from the fact that the outputs of the wind producers are strategic substitutes because $P'<0, P''\leq0$. \highlighta{Therefore, the best reply for firm $i$ is decreasing in firm $j$'s production, and firm $i$'s best response is a decreasing function of $\mathbb{E}[q_j | w_i = H]$.} When $d$ increases, the possibility that the wind producers are in different state increases. Thus,  the probability  that firm $j$ is in the low state increases, given that firm $i$ is in the the high state, and $\mathbb{E}[q_j | w_i = H]$ decreases, increasing $\phi$ which is firm $i$'s optimal production when it is in the high state.  

\begin{Lemma}\label{lem2}
Each firm's production in the high state (i.e. $\phi$) and the expected value of total production $\E_{w_1,w_2}(Q)$ both change in the same direction as a result of changes in dispersion. That is, 
$sign(\frac{\p}{\p d}\E_{w_1,w_2}(Q)) = sign(\frac{\p \phi}{\p d})$.
\end{Lemma}

\highlighta{The \emph{a priori} expected value of firm $i$'s production is $\E_{w_1,w_2}(q_1) = \beta \phi + (1-\beta) L$, and the expected value of total production is just the sum of each firm's expected production: $\E_{w_1,w_2}(Q) = \E_{w_1,w_2}(q_1) + \E_{w_1,w_2}(q_2) = 2 \beta \phi + 2(1-\beta) L$. Only $\phi$ on the right-hand side is a function of $d$; the parameter $d$ has no effect on the prior probability of being in the high state $\beta$. Then $(\frac{\p}{\p \phi}\E_{w_1,w_2}(Q)) = 2 \beta \frac{\p \phi}{\p d}$, with $\beta > 0$, which concludes the proof.}

\bigskip

\highlight{\noindent{{\bf Introducing Strategic Curtailment (SC) and Wind Diversification (WD)}.}} These two features describe the effects of $d$ on, respectively, high state output $\phi$ and on the joint probability distribution of the resource availability amongst all producers. The main effects of $d$, for instance on welfare, are driven by its impacts on strategic curtailment and wind diversification.  

\begin{itemize}
\item {\bf Strategic Curtailment (SC)}: When $d$ increases it impacts the information available to the wind producers as strategic decision makers. As a result, as $d$ grows, the production of firm $i$ when they are in the high state increases (Lemma \ref{lem1}). \highlighta{Equivalently, this increases the expected value of production (Lemma \ref{lem2}), and decreases the level of {\bf strategic curtailment}, the difference between the expected value of availability and the expected value of production,} i.e. $\mathbb{E}[w_i - q_i]$,. Thus when $d$ grows the level of {\bf strategic curtailment} decreases because increasing $d$ reduces the information quality available to producers in the high state and therefore reduces their strategic withholding of available production. 

\item {\bf Wind Diversification (WD)}: When $d$ grows the probability of being in different states increases. Consequently, with increasing $d$ firms produce different quantities with a higher probability, improving diversification of the total portfolio of wind producer assets and reducing the variance of the total availability of wind energy $\Var(w_1+w_2)$. In cases where utility is strictly concave, diversification of wind assets increases welfare.
\end{itemize}
The impact of {\bf wind diversification } on a function is measured as follows. Let $f:R^2 \rightarrow R$.  The impact of diversification on $f$ (denoted  by $WD_f$) is given by the following expression.
Let $y>x>0$ then
\[
WD_f\equiv f(x,y)+f(y,x)-f(x,x)-f(y,y).
\]

If $f$ is linear, i.e. $\exists$ $a,b,c \in R$ such that $f(x,y)=ax+by+c$, then $WD_f=0$.
Thus, the effects of wind diversification are inactive for linear functions. Furthermore, if $f$ is a concave function of the sum of its arguments, i.e. if $f(x,y)=f(y,x)=g(x+y)$ for some $g:R\rightarrow R$ where $g''<0$ then $WD_f=2g(x+y)-g(2x)-g(2y)>0$. 

Many of the results presented here are due to the interplay between the effects of $d$ on strategic curtailment and diversification as introduced above. In general, increasing $d$ improves social welfare through its effects on both diversification and strategic curtailment. However, because increasing $d$ decreases strategic curtailment, it can sometimes reduce profits for wind producers. This suggests that the level of heterogeneity sought by profit-maximizing investors can be lower than the welfare-maximizing level.    

\section{Social welfare vs. Dispersion} \label{Sec::welfare}
\highlighta{Since the marginal cost of energy production from wind is negligible, welfare (i.e. firms' surplus plus consumers' surplus)  is equivalent to the consumers' net utility of consumption. Let $U(Q)$ be the consumer utility, where $U(0)=0$ and $\forall Q$, $U'(Q) >0, U''(Q)\leq 0$. Note that $U’(Q)$ defines the inverse demand $P(Q)$. The consumer surplus is given by $CS=U(Q)-Qp$, and welfare is $W=\pi_1+\pi_2+CS=U(Q)$.}

\begin{Proposition} \label{prop2}
The expected value of welfare increases in dispersion $d$. 
\end{Proposition}

\highlighta{The expected value of welfare is given by $\mathbf{E}_{w_1,w_2}[W]=\mathbf{E}_{w_1,w_2}[U(q_1+q_2)]$. By the product rule of differentiation, the total impact of $d$ on welfare is exactly the sum of its impacts on $\mathbb{E} [W]$ through strategic curtailment and wind diversification. Increasing $d$ increases the expected value of welfare because it decreases strategic curtailment and increases diversification, which {\bf both} increase $U$.}

\highlighta{Increasing $d$ reduces strategic curtailment: it increases $\phi$, as shown in Lemma \ref{lem1}. This increases $q_i$ whenever $w_i = H$, which also increases $U(Q)$ because $U'>0$. Increasing $d$ also increases wind diversification: it increases the probability that wind producers are in different states. This increases the probability that $Q$ takes on its middle value, and decreases the probability that it takes on an extreme value. Since $U$ is concave, the impact of diversification on $U$ is weakly positive, i.e. $WD_U\geq0$, as shown above. Figure 1 illustrates these effects.  This intuition becomes clear with the following proof.}
\begin{proof}
Since $W=U(q_1+q_2)$, the expected value of social welfare is given by:
\begin{align*}
\E_{w_1,w_2}[W]=\Pr\{L,H\}U(L+\phi)+\Pr\{L,L\}U(2L)+\Pr\{H,H\}U(2\phi)+\Pr\{L,H\}U(L+\phi).
\end{align*}
In addition, as shown in the appendix (see \eqref{chance}),  $-\frac{\p \Pr\{L,L\}}{\p d}=-\frac{\p \Pr\{H,H\}}{\p d}=\frac{\p \Pr\{L,H\}}{\p d}\equiv \zeta>0$. That is, when $d$ increases chances for being in different states increases.  So,
\begin{align*}
\frac{\p }{\p d}\E_{w_1,w_2}[W]&=\underbrace{\underbrace{\zeta}_{>0} \underbrace{(2U(L+\phi)-U(2L)-U(2\phi))}_{=WD_U \geq 0, \text{\ by concavity of $U$}}}_{\text{\bf wind diversification}}\\
&\quad + 2 \ \underbrace{\frac{\p \phi}{\p d}}_{\substack{>0, \ \text{\bf reduction of} \\ \textbf{\bf strategic curtailment}}} \ \left ( \underbrace{\Pr\{L,H\}P(L+\phi)+\Pr\{H,H\}P(2\phi)}_{>0}\right).
\end{align*}
 As shown in \eqref{chance}, $\zeta>0$ because  the chance of producers being in different states increases in $d$. Furthermore,  concavity of $U$ implies that $WD_U=2U(L+\phi)-U(2L)-U(2\phi)>0$. Thus, wind diversification has  a positive impact on welfare. 
 
In addition,  by increasing $d$, the production in the high state increases; $\frac{\p \phi}{\p d}>0$ by Lemma \ref{lem1}. Therefore, the reduction of strategic curtailment (due to increasing $d$), has a positive impact on welfare. Overall, increasing dispersion increases the expected value of social welfare. Figure \ref{fig :: welfare2} shows these effects.    

 \begin{figure}
 \begin{center}
\includegraphics[width=0.8\textwidth]{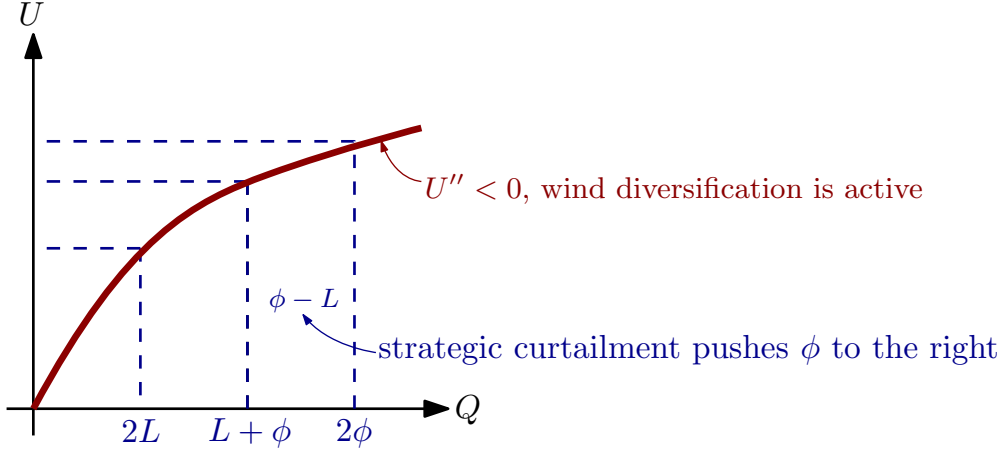}
\end{center}
\caption{Wind diversification increases $U$ and information effects decrease strategic curtailment which also increases $U$.}
\end{figure} \label{fig :: welfare2}

\end{proof}

\section{Price and Profit vs. Dispersion} \label{Sec::price}
How does extent of heterogeneity/dispersion affect average price and profit? We show the effect in general is {\it ambiguous} because the impacts of diversification and of changing levels of strategic curtailment on average price and profit are \emph{not} aligned. To understand this, we first analyze how average price responds to changes in dispersion. Figure \ref{fig :: price} shows these effects.  

\begin{Proposition}

\highlighta{The general impact of dispersion $d$ on the expected value of price is ambiguous. In the case of linear inverse demand, increasing $d$ decreases the expected value of the price.}
\end{Proposition}

\begin{proof}
Let $P''<0$. Since $$\E_{w_1,w_2}[P(q_1(w_1)+q_2(w_2))]=2\Pr\{L,H\}P(L+\phi)+\Pr\{L,L\}P(2L)+\Pr\{H,H\}
,$$ thus
\begin{align*}
\frac{\p }{\p d}\E_{w_1,w_2}[P]&=\underbrace{\underbrace{\zeta}_{>0} \underbrace{(2P(L+\phi)-P(2L)-P(2\phi))}_{WD_P>0, \text{\ by strict concavity of $P$}}}_{\text{\bf wind diversification}}\\
&\quad + 2 \ \underbrace{\frac{\p \phi}{\p d}}_{\substack{>0, \ \text{\bf reduction of} \\ \textbf{\bf strategic curtailment}}} \ \left ( \underbrace{\Pr\{L,H\}P'(L+\phi)+\Pr\{H,H\}P'(2\phi)}_{<0, \ \text{$P$ is downward}}\right).
\end{align*}

Higher dispersion reduces strategic curtailment, which decreases the average price because inverse demand is downward, i.e. $P'<0$. However, diversification increases the average price because of concavity in inverse demand, i.e. $WD_P=2P(L+\phi)-P(2L)-P(2\phi)>0$. \highlighta{The net effect is ambiguous.}  

\highlighta{Note that when inverse demand is linear, i.e. $P''=0$, then $WD_P=2P(L+\phi)-P(2L)-P(2\phi)=0$. Thus, the effect of diversification is completely inactive.  As a result, because of the impacts of $d$ on strategic curtailment, the expected value of the market price decreases in $d$ in the case of a linear inverse demand.}

\end{proof}

 \begin{figure}
 \begin{center}
\includegraphics[width=0.9\textwidth]{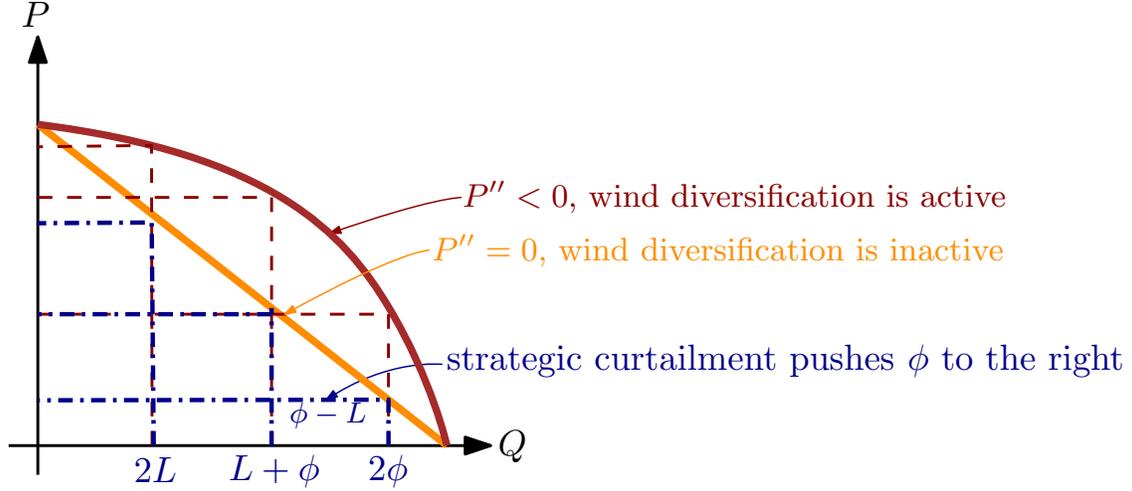}
\end{center}
\caption{Interplay between the effects of wind diversification and strategic curtailment on average price. Wind diversification increases the average price when $P''<0$ and is inactive when $P''=0$. The impacts of increasing $d$ on strategic curtailment always decrease the average price.}
\end{figure} \label{fig :: price}


\highlighta{Like the average price, the impact of increasing dispersion on profit is in general ambiguous. When $d$ increases, it decreases $\phi$. This increases profit when $w_1 = w_2 = H$ because $\phi$ is greater than the full-information Cournot output. However, it decreases the average profit when $w_1 \neq w_2$ because $L + \phi$ is less than the sum of full-information Cournot outputs. Increasing $d$ also increases the probability that the two producers have different resource availability, $\Pr\{w_1 \neq w_2\}$, which increases the expected value of profit because diversification has a positive effect on profit. We can again characterize the effect of $d$ on profit completely through its effects on strategic curtailment and diversification. The overall impact of dispersion on profit is ambiguous, as shown in the following Example.} 

\begin{Example}\label{th-33}
Let $P'<0$, $P''<0$. As $d$ increases, the expected value of producer profit increases due to diversification and decreases due to reduced strategic curtailment. Thus, in general, the impact of heterogeneity on profit is ambiguous.
\end{Example}

 \begin{figure}[h]
 \begin{center}
\includegraphics[width=0.85\textwidth]{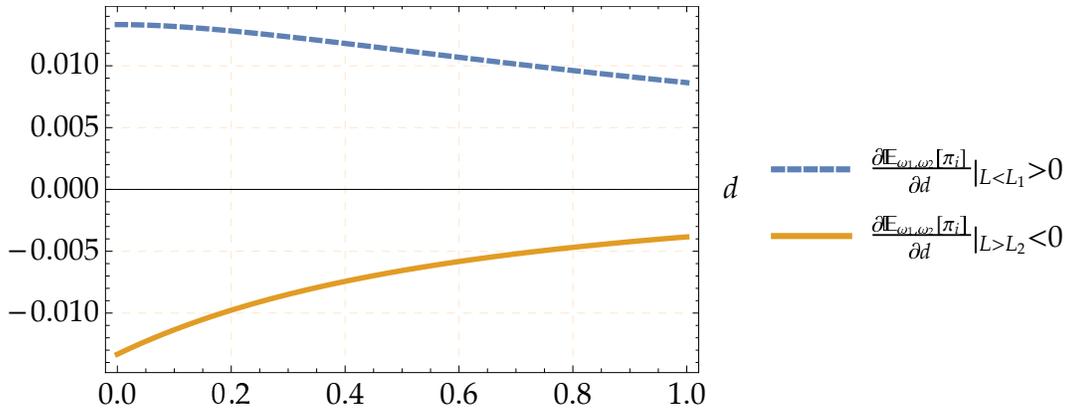} 
\end{center}
\caption{Wind diversification (heterogeneity) increases profit if $L$ is sufficiently small, and it  decreases profit if $L$ is sufficiently large. Plot parameters: $s=3, \beta=\frac{1}{2}$, for dashed line $L=0.6$ and for the solid line $L=0.8$.}
\end{figure} \label{fig:profit}

\highlighta{In general, increasing dispersion $d$ can increase or decrease the expected value of profit. However, in the case of linear inverse demand, we can obtain sharp insights based on the absolute value of the low state energy availability $L$. This is because the extent of $L$ affects the \emph{strength} of diversification and changing strategic curtailment levels on profit in \emph{opposite} directions. As such, for sufficiently high $L$, increasing heterogeneity $d$ reduces profits. See Figure \ref{fig:profit}. The following Proposition summarizes:}

\begin{Proposition}\label{th-firms}
\highlighta{Let $P(q_1+q_2)=s-q_1-q_2$, then there exists two thresholds $L_1$ and $L_2$, with $L_1 = \frac{2s}{9} < L_2 = \frac{2s}{8}$, such that}
\begin{itemize}
\item[(i)] When $L<L_1$ the impact of diversification dominates the strategic curtailment effects, thus $\frac{\p}{\p d}\E_{w_1,w_2}[\pi_i]>0$. Consequently,  it is beneficial for firms to place their wind farms \emph{far} from each other, i.e.  
\begin{align*}
\arg \max_{d\in [0,1]} \ \E_{w_1,w_2}[\pi_i]= 1.
\end{align*}
\item[(ii)] When $L>L_2$,  strategic curtailment dominates diversification, thus $\frac{\p}{\p d}\E_{w_1,w_2}[\pi_i]<0$. Consequently, it is beneficial for firms to place their plants \emph{close} to  each other, i.e.
\begin{align*} 
\arg \max_{d\in [0,1]} \ \E_{w_1,w_2}[\pi_i]= 0.
\end{align*}
\end{itemize}
\end{Proposition}

\highlighta{When $L$ is sufficiently high, Proposition \ref{th-firms} shows that the expected value of profit for each producer is decreasing in $d$. In this case, producers prefer lower $d$ even though higher $d$ improves overall welfare, as shown in Proposition \ref{prop2}. This suggests that profit and welfare motives may sometimes be misaligned, since dispersion uniformly improves social welfare but may not improve profit. For example, a regulator may propose policies to increase $d$ by encouraging investment far from existing sites, but firms might oppose such policies because they reduce the expected value of profits. }



\section{Multiple Wind Generators with a Generic Joint Distribution of Wind Availability}\label{Sec::multw}

This section shows that the main results of the paper extend to markets with multiple wind generators. We demonstrate a parsimonious way to extend the notion of dispersion $d$ to markets with an arbitrary number of wind producers, and we show that high state output and welfare are still increasing in $d$ due to its effects on strategic curtailment and diversification.

Consider a market with $N+1$ wind generators, each with probability $\beta$ of producing in the high state and separated by dispersion $d$. Here, $d$ gives a proxy for the level of correlation among the \highlighta{states of different producers, where as before high $d$ implies that the stochastic resource availabilities of different producers are more independent}. \highlighta{We define the state of producer $i$ as $s_i = \mathbb{1}_{\{w_i = H\}}$; this represents the state of producer $i$ as a 0 or a 1 (with $H$ corresponding to 1). Let $S_{-i} = \sum_{j \neq i} s_j$, the number of producers besides producer $i$ that are in the high state. Let $S = \sum_{i} s_i$, the total number of producers in the high state.} 

Then, consider the random vector $\vec{s}^d$ for $d \in (0,1]$ and assume $\beta > 0$. The probability distribution of $\vec{s}^d$ is the joint probability distribution $\Pr\{s_1, s_2, ..., s_{N+1}; d\}$.  Each of $s_i$ are random variables, as are the sums, and therefore both
\begin{align} 
& \Pr\{S = j ; d\} \quad j \in \{0,...,N+1\} \nonumber \\ 
& \Pr\{S_{-i} = k | w_i = H; d\} \quad k \in \{0,...,N\} \quad i \in \{1,....,N+1\} \label{cond}
\end{align}
are well defined. Moving forward, we use $S^d$ and $S_{-i}^{d}$ as the random variables of the sum of states generated by distributions parameterized by dispersion $d$. \highlighta{Again, we assume that distributions are symmetric; the probability law for $S_{-i}^{d}$ is equal to the probability distribution for $S_{-j}^{d}$ for all $i,j \in \{1,...N+1\}$.}

As before, we assume that $L$ is sufficiently small such that producers never curtail in the low state, i.e. $P((N+1)L) + L P'((N+1)L) > 0$. \highlighta{This is equivalent to the first part of Assumption \ref{assumption-11} in the duopoly case.}  
The first order optimality condition for $\phi$ is given by
\begin{equation} \label{multFOC}
\mathbb{E}_{S_{-i}}[P(\phi+S_{-i}\phi + (N-S_{-i})L) + \phi P'(\phi + S_{-i}\phi + (N-S_{-i})L)|w_i=H] = 0,
\end{equation}
\highlighta{where the expectation is evaluated using the conditional probability distribution in \eqref{cond}. Furthermore, we assume that there exists some value $v < H$ that solves \eqref{multFOC} when $\phi = v$. This corresponds to the second part of Assumption \ref{assumption-11} for the oligopoly case, but it is a weaker requirement. It simply ensures that the equilibrium is of interest; otherwise, $q_i(w_i) = w_i$ and players always produce their full energy availability.}   

\highlighta{Under these assumptions, there is a unique root $\phi$ that solves \eqref{multFOC}, with $L<\phi<H$, and the unique symmetric BNE is given by $q_i(w_i) = \min\{w_i,\phi\}$. We adopt these assumptions for the remainder of this section, and let $\phi$ refer to the unique root of \eqref{multFOC}.}


Next we characterize two sufficient conditions on the effect of the parameter $d$ on the joint and conditional distributions.\footnote{These conditions are based on first- and second-order stochastic dominance, see \cite{shaked2007stochastic}.}
These conditions allow for the extension of our results on strategic curtailment and welfare to any arbitrary inverse demand curve with $P' < 0$, $P'' \leq 0$ in a market with $N+1$ producers.  
\begin{Assumption}\label{assumption-FOSD}
	\highlighta{For all $d' > d$, for each $i$, conditional on $w_i = H$,  $S_{-i}^{d} \succeq_{FOSD} S_{-i}^{d'}$} That is, $\forall i$, $\forall j \in \{0,...,N\} $
	\begin{equation*}
	 \Pr\{S_{-i} > j | w_i = H; d\} \geq \Pr\{S_{-i} > j | w_i = H; d'\}.
	\end{equation*}
\end{Assumption}
\begin{Assumption}\label{assumption-SOSD}
	For all $d' > d$,  $S^{d'} \succeq_{SOSD} S^d$ That is, $\forall m$, 
	\begin{equation*}
	\sum_{j=0}^{m} \big(\Pr\{S > j ; d'\}  - \Pr\{S > j ; d\} \big) \geq 0.
	\end{equation*}
\end{Assumption}
From the perspective of a single producer $i$ in the high state, Assumption \ref{assumption-FOSD} requires that more competitors are likely to be in the high states when dispersion $d$ is lower. The idea is that when dispersion $d$ is small, producer $i$ being in the high state provides a stronger signal that competitors are also more likely to be in the high state.

Assumption \ref{assumption-SOSD} says that when $d$ is higher, the sum of wind availability has at least as high a mean and less weight in the tails of its distribution. When $d$ is high, the resource availabilities of different producers are nearly independent. When $d$ is low, there is high correlation between producers, so there is a greater chance that a large number of producers ($\gg N/2$) are either in the high state or the low state. Both Assumptions \ref{assumption-FOSD}  and \ref{assumption-SOSD} are satisfied by the duopoly model in Section \ref{Sec::Model}.\footnote{Consider the duopoly model in Section \ref{Sec::Model}, but with general probability distributions $\Pr\{w_i=H|w_j=H\} \equiv f(d,\beta)$ and $\Pr\{w_i=H|w_j=L\} \equiv g(d,\beta)$. In Section \ref{Sec::Model}, specific functional forms are provided in \eqref{prob} for $f(d,\beta)$ and $g(d,\beta)$ in order to motivate the exposition. For generic conditional probabilities in the duopoly model, Assumption \ref{assumption-FOSD} establishes that $f(d,\beta)$ is weakly decreasing in $d$. Assumptions \ref{assumption-FOSD} and \ref{assumption-SOSD} together establish that $g(d,\beta)$ is weakly increasing in $d$.}

\begin{Proposition} \label{th-gen_SC}
	For general $N\geq 1$, given Assumption \ref{assumption-FOSD}, the output of producers in the high state $\phi$ is \highlighta{(weakly)} increasing in $d$. Therefore, as in the duopoly case, increasing $d$ (weakly) decreases strategic curtailment.  
\end{Proposition}
The left hand side of the first order condition \eqref{multFOC} is in general decreasing in the output of other producers. The intuition is that the expected value of the output of other producers, with $\phi$ fixed, is decreasing in $d$, from the perspective of a producer whose output is high. Therefore, higher $d$ increases the left hand side of \eqref{multFOC}. Lower $\phi$ also increases the left hand side. Thus, as $d$ increases, a lower $\phi$ cannot possibly solve the first order condition because both higher $d$ and lower $\phi$ increase the left hand side of \eqref{multFOC}.  
 
\begin{Proposition} \label{th-gen_welfare}
	For general $N\geq 1$, given Assumptions \ref{assumption-FOSD}, \ref{assumption-SOSD}, and $P(\phi^1 (N+1)) \geq 0$,\footnote{The variable $\phi^1$ represents the high state production when $d = 1$. This assumption implies that equilibrium prices will not become negative. In practice, we see negative prices arise in markets with high penetrations of wind resources, but this is due to the presence of subsidies, and non-convexities associated with traditional generation, not a result of producer strategy in the face of uncertainty.} the expected value of welfare $\mathbb{E}_{\vec{s}^d}[W]$ is increasing in $d$.  
\end{Proposition}

\highlighta{Consider $d' > d$. We aim to show that $\mathbb{E}_{\vec{s}^{d'}}[W] > \mathbb{E}_{\vec{s}^d}[W]$. Let $\phi^d$ refer to the equilibrium curtailment level as described by \eqref{multFOC}, for the random availability vector $\vec{s}^d$ indexed by $d$. Consider a given realization of the resource availability for each producer, and let $S = \tilde{s}$, for some $\tilde{s}$ with $0 \leq \tilde{s} \leq N+1$. We can describe welfare as a function $W(\tilde{s},\phi)$. The full proof in the Appendix explains that under the first-order conditions of the equilibrium described by \eqref{multFOC}, $W$ is concave and increasing in $\tilde{s}$. Welfare $W$ is also increasing in $\phi$.} 

\highlighta{The distributions of $S$ satisfy Assumption \ref{assumption-SOSD}, so the distribution with higher $d$ has total wind availability $S$ that second-order stochastically dominates the original distribution. The definition of second-order stochastic dominance implies that the dominating random variable leads to higher expected value for increasing concave functions. Therefore, holding $\phi$ constant, higher $d$ increases the expected value of welfare: $\mathbb{E}_{\vec{s}^{d'}}[W(\cdot, \phi^{d'})] > \mathbb{E}_{\vec{s}^{d}}[W(\cdot, \phi^{d'})]$. Furthermore, using Assumption \ref{assumption-FOSD}, Proposition \ref{th-gen_SC} shows that $\phi$ is increasing in $d$. Since $W$ is increasing in $\phi$, $\mathbb{E}_{\vec{s}^{d}}[W(\cdot, \phi^{d'})] > \mathbb{E}_{\vec{s}^{d}}[W(\cdot, \phi^{d})]$; together, the two inequalities establish that $W$ is increasing in $d$.}

\section{\highlight{Competition with Traditional Generation}} \label{sec :: trad}
This section considers Cournot competition between two wind producers and a traditional generator. The wind producers with dispersion $d$ and availability $\beta$ have identical characteristics to those described in Section \ref{Sec::Model}. The traditional producer can output any quantity $x \in \mathbb{R}^+$ with constant marginal cost $c \geq 0$; it has no information regarding the availability of the wind generators. As before, the inverse demand function $P(q_1,q_2,q_3)$ is a function of the sum of its arguments; with slight abuse of notation, we also write this as $P(q_1 + q_2 + q_3)$. 

\highlighta{This section extends the existing results on the impact of $d$ on welfare. As before, welfare is increasing in $d$. The models used in this section and Section \ref{Sec::multw} could be used to analyze markets with multiple wind producers and multiple traditional generators, but the analysis in this section is sufficient to highlight the main insights.} 
	
\highlighta{Towards our equilibrium assumptions, let $\ubar{x} \geq 0$ be the solution to $\mathbb{E}[P(w_1,w_2,\ubar{x}) + \ubar{x} P'(w_1,w_2,\ubar{x})] = c $. This value represents the lower bound for the total energy production by the traditional generator in an equilibrium.}  

\begin{Assumption} \label{Assumption_22}
\highlighta{Let $P(\cdot)$ be inverse demand and $c$ be the marginal cost of using traditional generation. 
Assume $c < P(2H)$, which guarantees that the traditional generator produces a positive quantity of energy. Assume $P(3 L) + LP'(3 L)>0$ and $P(H + L + \ubar{x})+HP'(H + L + \ubar{x})<0$.
} 

\end{Assumption}
\highlighta{The assumptions extend Assumption \ref{assumption-11} to the case of three players with a traditional generator. They guarantee that we have a solution of interest, so we avoid explaining the cases whereby $L$ is sufficiently high that wind producers might always curtail, where $H$ is too low so that wind producers might never curtail, or where $c$ is sufficiently high that the traditional producer will never produce.}\footnote{The assumption establishes an upper limit on $c$. When $c$ is lower, the output of the traditional generator increases because they have lower marginal costs of production.}



\begin{Proposition} The Cournot equilibrium for generic inverse demand $P(\cdot)$, with $P'<0$, $P''<0$ satisfies the following first order conditions, where \eqref{A} is the first order condition for wind producers and \eqref{B} is the first order condition for the traditional producer.
\begin{equation} \label{A}
\Pr\{L|H\} (P(L+\phi+x) +\phi P'(L + \phi + x)) + \Pr\{H|H\} (P(2 \phi+x) +\phi P'(2 \phi + x)) = 0 
\end{equation}  
\begin{multline} \label{B}
\Pr\{L,L\} (P(2 L+x) +x P'(2 L + x)) + 2 \Pr\{L,H\} (P(L+ \phi+x) + x P'(L+\phi + x)) \\ + \Pr\{H,H\} (P(2 \phi+x) + x P'(2 \phi + x)) - c'(x) = 0 
\end{multline}
\end{Proposition}
The result follows exactly from Proposition \ref{pro0} with the addition of the traditional generator whose output satisfies the first order condition described in \eqref{B}. \highlighta{Equation \eqref{A} describes the equilibrium high state output $\phi$ for wind producers to maximize their profit, contingent on the equilibrium behavior of the other wind producer and the traditional generator. Equation \eqref{B} describes the equilibrium output $x$ of the traditional generator, with $c'(x) = c$ in our example.}  
\begin{Example} \label{prop :: trad_gen eq}
	Consider a market with linear inverse demand, $P(q_1,q_2,q_3) = s - q_1 - q_2 - q_3$. Then the unique high state output of the wind generators $\phi$ and the production output of the traditional generator $x$ are given by:
	\begin{equation} \label{phi}
	\phi = \frac{\frac{1}{2}(s+c)(\beta + d(1-\beta))+ L \beta(1-\beta)(1-d)}{3\beta + 2d(1-\beta) - \beta^2 - \beta d (1-\beta)} 
	\end{equation}  
	\begin{equation} \label{wphi}
	x = \frac{1}{2}(s-c) - \phi \beta - L (1-\beta).  
	\end{equation}
\end{Example}
\noindent The Example is explained in the Appendix. It is obtained by solving the equations \eqref{A} and \eqref{B} in terms of $\phi$ and $x$ in the case of linear inverse demand.  

Next we consider the impact of heterogeneity on strategic curtailment by wind producers $\frac{\partial \phi}{\partial d}$ and quantity withholding by the traditional producer $\frac{\partial x}{\partial d}$ in the case of linear inverse demand. 


\begin{Corollary} \label{trad_linearchanges}
	\highlighta{Consider a market with two wind producers indexed by $i=\{1,2\}$ and one traditional producer. Let inverse demand $P$ be linear with $P(q_1,q_2,q_3) = s - q_1 - q_2 - q_3$. Then $\phi$ and $\mathbb{E}[q_i]$ are increasing in $d$ for and the output $x$ of the traditional generator is decreasing in $d$.}
\end{Corollary}

\begin{proof}
	We can take the derivative of $\phi$ with respect to $d$, using the form of the equation in \eqref{phi}. 
	\begin{equation} \label{dphidd}
	\frac{\partial \phi}{\partial d} = \frac{(s+c-4 L)\beta(1-\beta)}{2 (3\beta + 2d(1-\beta) - \beta^2 - \beta d (1-\beta))^2} 
	\end{equation}
	Under our assumptions, this is always positive. Equation \eqref{dphidd} is always positive when $s+c-4L >0$, which is always satisfied by Assumption \ref{Assumption_22}. Therefore, the output of the wind generators is increasing in $d$, $ \frac{\partial \phi}{\partial d} > 0$, as in the original market.     
	
	Then, taking the derivative of $x$ using the first order condition in \eqref{wphi}, $\frac{\partial x}{\partial d} = -\beta \frac{\partial \phi}{\partial d} < 0$. Therefore, the output of the traditional generator is decreasing in $d$, so the traditional generator withholds more when the wind generators have less information about the other wind producer's state.
\end{proof}

 \begin{figure}[h]
 \begin{center}
\includegraphics[width=0.6\textwidth]{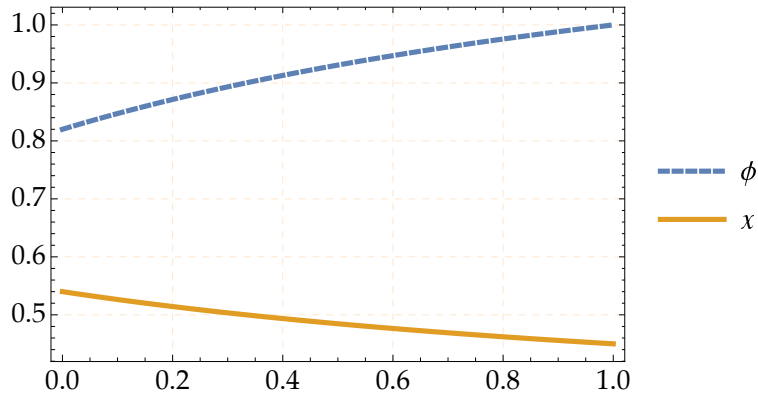}
\end{center}
\caption{This chart shows that the traditional firm's output $x$ decreases in the diversification $d$, but  the average output (and the high state production $\phi$) of the wind generators is increasing in $d$.   Plot parameters: $s=3, \beta=\frac{1}{2}, L=0.1$ and $c=1$.}\label{fig:phix}
\end{figure}

\subsection{Welfare vs. Dispersion in a Market with Fossil-Fuel Generation}

Next, we consider the effects of heterogeneity on welfare. Increasing dispersion $d$ improves welfare in the market that includes a traditional generator.

\begin{Proposition} \label{prop :: trad_gen welfare}
	\highlighta{In the three player market with two wind producers and a traditional producer, and a linear inverse demand $P(q_1,q_2,q_3) = s - q_1 - q_2 - q_3$, the expected value of welfare is increasing in dispersion $d$.}
\end{Proposition}



\highlighta{In this model, as before, increasing $d$ still reduces the strategic curtailment of wind producers $\frac{\partial \phi}{\partial d} > 0$, and improves wind diversification. However, when a fossil fuel generator has market power, the fossil fuel generator responds by withholding more of their own output due to strategic substitutability with $\mathbb{E}[q_1 + q_2]$, which increases; thus $\frac{\partial x}{\partial d} < 0$. With a linear inverse demand, the FOCs \eqref{A} and \eqref{B} imply that the sum of the welfare impacts, due to the changes in the equilibrium values of $\phi$ and $x$, is 0. Thus, increasing $d$ only impacts the expected value of welfare through the change in wind diversification, which positively impacts welfare. } 


\subsection{Price vs. Dispersion in a Market with Fossil-Fuel Generation}
Finally, we show that in a market with traditional generation and linear inverse demand, $d$ decreases the \highlighta{expected value of price}. This result extends earlier results to the case of a market with traditional generators and highlights the potential benefits of increased heterogeneity for reducing market prices. 

\begin{Proposition} \label{prop :: trad_gen price}
	The \highlighta{expected value of the market price, given by $\mathbb{E}_{w_1,w_2}[P(q_1(w_1)+q_2(w_2) + x)]$, satisfies $\frac{\partial \mathbb{E}_{w_1,w_2}[P]}{\partial d} = -\beta \frac{\partial \phi}{\partial d} <0$.}
\end{Proposition} 

\highlighta{ The expected value of total energy production is increasing in $d$. Its derivative with respect to $d$ is given by $2 \Pr\{H\} \frac{\partial \phi}{\partial d} + \frac{\partial x}{\partial d} = (2\beta  - \beta) \frac{\partial \phi}{\partial d} > 0$, where the equality is because $\frac{\partial x}{\partial d} = \beta \frac{\partial \phi}{\partial d}$, as explained in the proof of Corollary \ref{trad_linearchanges}. Under linear inverse demand, the expected value of the market price is decreasing in $d$. }

Since the production by the traditional generator is uniformly decreasing in $d$, increased dispersion reduces profits for the traditional generator. On the other hand, the effects of $d$ on wind producer profits are ambiguous, as was the case in the original model. 

\section{\highlight{Collusion, Prevention, and Dispersion}} \label{Sec:collusion}
This section investigates potential collusion between wind producers and studies the effect of increased heterogeneity on the presence of collusion. It focuses on linear inverse demand for simplicity, and shows that collusion is always possible (incentive compatible) among wind producers when there are no penalties for collusion. The section also examines the case where firms may face random penalties for engaging in collusion, so the threat of sanctions poses a random cost on the decision to collude.

Consider two wind producers that are willing to collude in order to increase profits. They set up a contract to produce at the monopoly level when possible and share profits depending on their stated availability.  The true availability of wind is private information, so a wind producer in the high state can lie about their state and produce the amount of wind appropriate for a producer in the low state.\footnote{\highlighta{We assume the the contract is enforceable with regards to production quantities, which are publicly verifiable. Therefore, if the producer announces that they are in the $H$ (or $L$) state, then in any equilibrium they will produce the agreed upon amount for a producer in that state, regardless of their true state. However, it is not possible for a firm to verify the true state of their competitor (which is private information); out of equilibrium, a producer could choose to lie about its production constraint.}}

Let $\pi_M$ be the combined monopoly profits and $\pi_L$ be the profits when both producers are in the low state. In the case of linear inverse demand, where $P(q_1,q_2) = s - q_1 - q_2$, the optimal output for a monopoly producer is $q_M = \frac{s}{2}$.\footnote{The implication is that colluding producers will produce such that their combined output is at the monopoly level when possible, i.e. $q_1 + q_2 = q_M$, and they will devise a mechanism to share the monopoly profits. } Then

\begin{equation*}
\pi_M = P(q_M) q_M = \frac{s^2}{4} \quad \quad \quad \pi_L = P(L,L) L = (s-2L) L
\end{equation*}

There is an exogenous cost to collusion $\gamma \geq 0$, to explain a situation where the government tries to identify and penalize collusion. \highlighta{We can think of $\gamma$ as being the government's penalty for a firm engaged in collusion, times the probability of detection. We can imagine that the government might undertake various efforts to identify collusion, for instance by reviewing price trends, measuring the difference between wind forecasts and outputs, or monitoring information exchange between competing firms.} 

\highlighta{Producers who collude jointly produce at the monopoly level and share the monopoly profits $\pi_M$ when they are both in the high state; in this case, they each receive $\frac{\pi_M}{2}$. They each produce $L$ and receive $\pi_L$ when they are both in the low state. Additionally, the producers set up a transfer scheme whereby producers that are in the high state pay $t \pi_M$ to producers that are in the low state.} The transfer fraction $t$ represents the fraction of monopoly profits given to the low state producer; since arbitrary $ t \in \mathbb{R}$, and $\pi_M > 0$, any real number is a feasible transfer; the total transfer is written as a product of $t$ and $\pi_M$ (as opposed to a single parameter) because it simplifies the exposition. 

Collusion is possible whenever there exists a monetary transfer $t \pi_M$ satisfying the incentive compatibility constraint (IC), which implies that high state producers will not lie and pretend they are in the low state, and which satisfies the \highlighta{individual} rationality (IR) constraints, which implies that producers will know \emph{ex-ante} that they would like to participate regardless of their unknown state. The incentive compatibility constraint (IC) is given by \eqref{IC}.
\begin{equation} \label{IC} 
\Pr\{H|H\}\frac{\pi_M}{2} + \Pr\{L|H\} (\pi_M - t\pi_M) \geq \Pr\{H|H\} t \pi_M + \Pr\{L|H\} \pi_L
\end{equation}
\highlighta{The IC constraint says that the expected value of the profit for a colluding producer $i$ in the high state is greater than the expected value of the profit they would receive if they lied and declared that they were in the low state, where the expected value of the profit is given using the conditional probabilities for their competitor $j$, given that $i$ is in the high state.}
The individual rationality constraints (IR) for high and low state producers are given by \eqref{IRH} and \eqref{IRL}, respectively.  
\begin{equation} \label{IRH}
\Pr\{H|H\}\frac{\pi_M}{2} + \Pr\{L|H\} (\pi_M-t \pi_M)  - \gamma \geq \Pr\{H|H\} \phi (s-2\phi) + \Pr\{L|H\} \phi (s - L - \phi)
\end{equation}
\begin{equation}\label{IRL}
\Pr\{H|L\}t \pi_M + \Pr\{L|L\} \pi_L - \gamma \geq \Pr\{H|L\} L (s-\phi - L) + \Pr\{L|L\} \pi_L
\end{equation}
The variable $\gamma$ is the expected cost from collusion, due, for instance to attempted prosecution by the government. \highlighta{As before, the conditional probability $\Pr\{H|L\}$ refers to $\Pr\{w_j = H| w_i = L\}$, the probability that the competitor $j$ is in the high state given that a wind producer $i$'s own availability is the low state; this is the same for other combinations of the states $H$ and $L$. Equation \eqref{IRH} explains that a producer in the high state would prefer to collude than to participate in the strategic equilibrium from Section \ref{sec :: equil}. Equation \eqref{IRL} explains that a producer in the low state would prefer to collude than to participate in the strategic equilibrium from Section \ref{sec :: equil}. Both of these constraints must hold; otherwise, a firm would not participate ex-ante because they would recognize that they would terminate the collusion agreement if they were in one state, revealing their availability to their competitor and reducing their profit.}  

\begin{Proposition} \label{collusion_poss}
	If there is no cost to collusion, i.e. $\gamma=0$, then there is always an available transfer satisfying the IC and IR constraints. \highlighta{That is, when $\gamma = 0$, $\exists t \in \mathbb{R}$ that satisfies \eqref{IC}, \eqref{IRH}, and \eqref{IRL}. Therefore, when $\gamma = 0$, producers can always increase profits by colluding.}
\end{Proposition}

The intuition is that a transfer is always possible when $\gamma = 0$ because the sum of profits from the generators strictly improves when they collude and when one producer is in the high state, so the benefit to producers in the high state is larger than the cost to producers in the low state, and thus there is a feasible transfer that allows collusion to be beneficial for producers \emph{ex-ante}. Next, we examine the effect of $d$ on various features of collusion.   

\noindent{{\color{Maroon}{\bf How does dispersion $d$ impact collusion?}}} \quad

\highlighta{In general, we find that dispersion $d$ does not have generalizable impacts on collusion in our model. Dispersion does not have monotonic impacts on the value of collusion to producers. It also does not monotonically impact the change in welfare due to collusion by producers. Figure \ref{fig :: collusion} summarizes these two effects.}  

\begin{figure}[h]
	\centering
	\begin{tabular}{c c c}
		\hbox
		{\subfloat[]{\includegraphics[scale = 0.4]{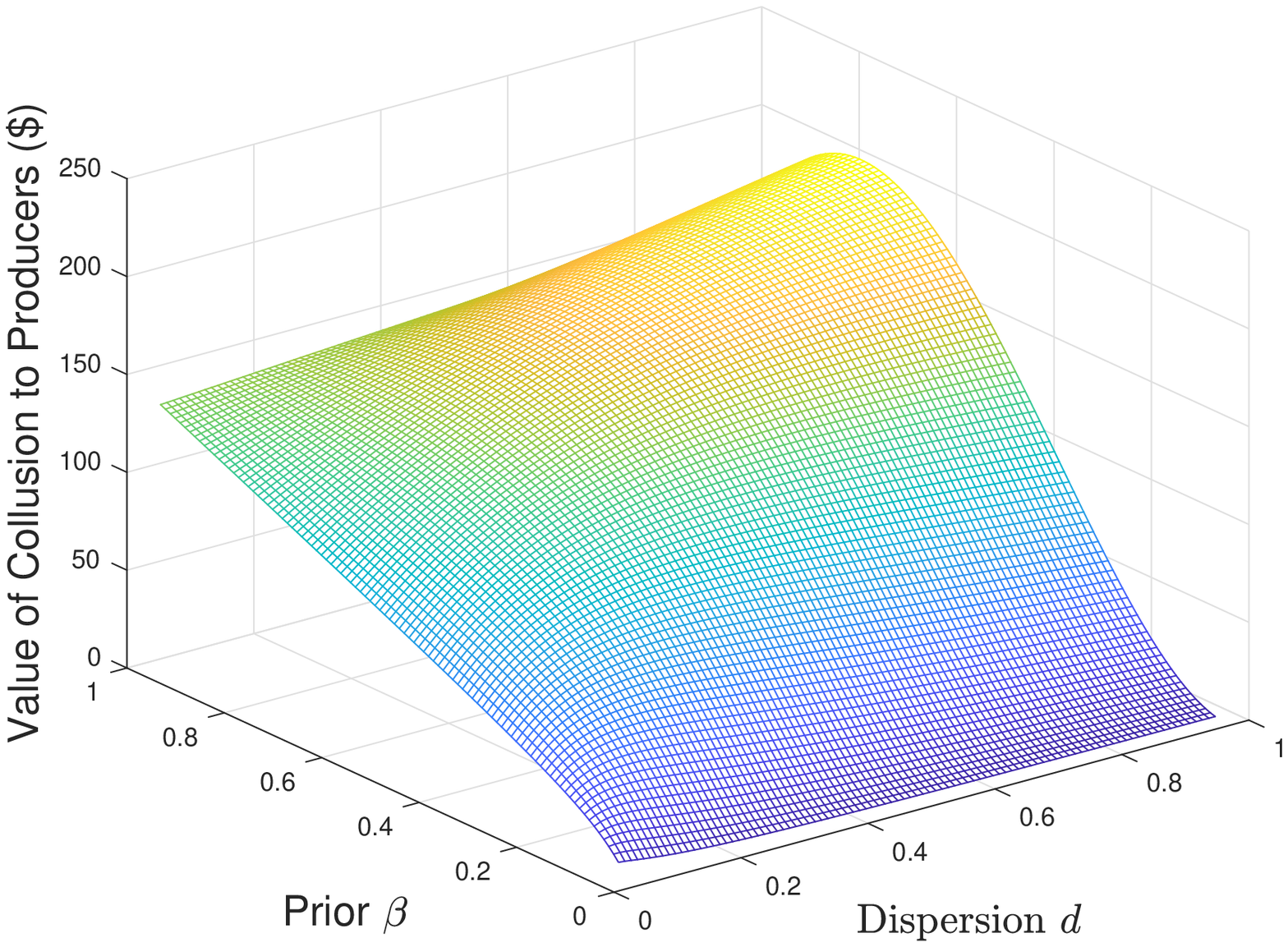}}} &
		{\subfloat[]{\includegraphics[scale = 0.4]{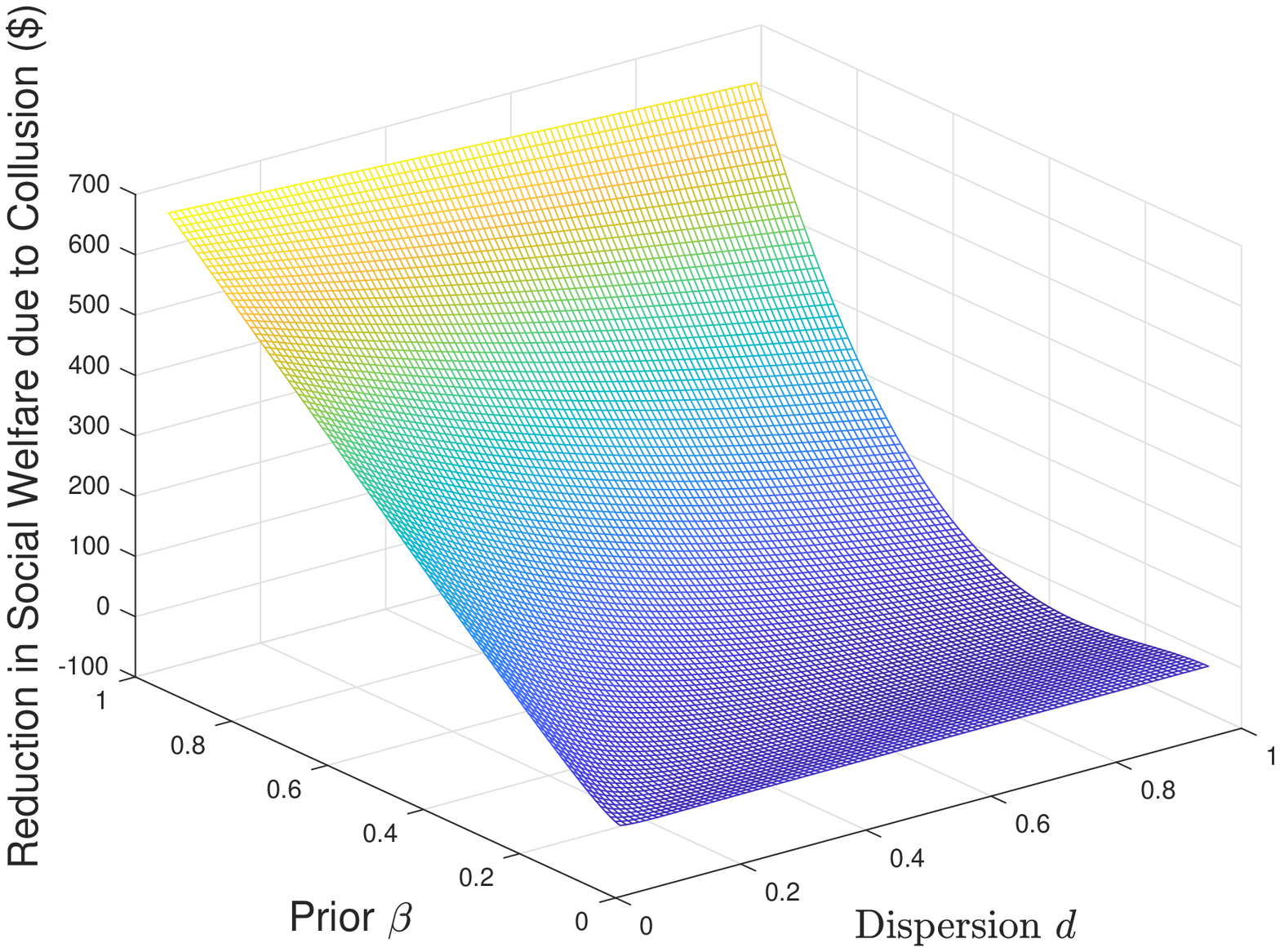}}} \\
	\end{tabular}
	\caption{The impact of dispersion on various features of collusion. (a) shows the impact of dispersion on the value of collusion to producers. (b) shows the impact of dispersion on the costs of collusion in terms of a reduction of social welfare.}
	\label{fig :: collusion}
\end{figure}

\highlighta{We can also estimate the impact of $d$ on policies intended to prevent collusion. Let $\hat{\gamma}$ represent the minimum $\gamma$ such that \eqref{IC}, \eqref{IRH}, and \eqref{IRL} imply a contradiction. We can say that $\hat{\gamma}$ represents the minimum expected value of a collusion penalty such that enforcement is sufficient to prevent collusion; if $\hat{\gamma}$ is very high, this implies that collusion must have a high probability of being punished and/or that the punishment must be very severe in order to prevent collusion. We find that dispersion does not have monotonic impacts on $\hat{\gamma}$. Figure \ref{fig :: collusion2} displays this effect.} 

\begin{figure}[h]
	\centering
	\includegraphics[scale = 0.6]{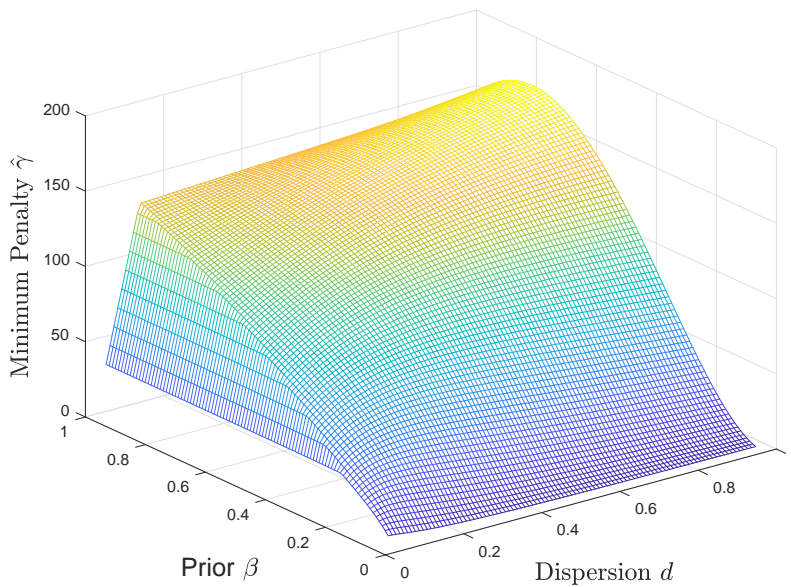}
	\caption{The impact of dispersion on the level of enforcement required to prevent collusion.}
	\label{fig :: collusion2}
\end{figure}

\section{Public Forecasting: Who Benefits?} \label{sec :: infosharing}

\highlight{This section focuses on the benefits of public sharing of information under the assumption that wind producers do not collude. It investigates the benefits of publicly providing high-quality short-term weather forecasts or real-time wind speeds for all wind-generating locations. It suggests that public forecasting always improves welfare, but it does not always benefit producers. This suggests that producers will not provide public forecasting, but that policy makers should fund forecasting efforts to improve the quality of public information.     
	
The results in this section show that information sharing always improves social welfare. However, we also show that when $L$ is sufficiently large (as a function of dispersion $d$), wind producers do not choose to share information. The limit on $L$ is increasing as a function of dispersion $d$. The results suggest that policies to implement public weather forecasting can improve welfare, because profit-maximizing producers will not always share weather information even though doing so always improves social welfare.} 

\highlight{In order to understand the effects of information sharing on social welfare and producer profit, we compare the baseline model (see Section \ref{Sec::Model}), where wind availability (i.e. state) is private information, to the case where wind energy producers {\bf ex-ante} commit\footnote{We assume wind producers are committed and there is no room for adverse selection. For instance, there could be automatic equipment for weather monitoring that shares information publicly.} to share their private information about their available energy, given that the extent of wind producer heterogeneity is $d\in (0,1]$. We assume inverse demand is linear, i.e. $P(q_1+q_2)=1-q_1-q_2$.} Under this assumption, the net welfare obtained by consuming $Q$ units of energy $U(Q) = \int_{0}^{Q} P(q) dq = \int_{0}^{Q} (1-q) dq = Q - \frac{1}{2} Q^2$.

\paragraph{{\color{Maroon}{\bf Is sharing information between wind producers \emph{socially} beneficial?}}}   Information sharing has both positive and negative effects on welfare. It helps prevent producers in the high state from inefficiently withholding their output when the other producer is in the low state, but it also introduces additional costs to welfare due to the reduction in welfare when producers producer at the Cournot output when they are both in the high state. In general, however, these impacts are in favor of the benefits of sharing information, as the following proposition summarizes. 

\begin{Proposition}\label{th-6}
	Sharing information between wind producers  is  \emph{socially} ex-ante beneficial.
\end{Proposition} 
Throughout this section, we let $K$ denote the equilibrium outcomes when wind producer share private information (or that information is made public), and we let $K^c$ denote equilibrium outcomes when producers compete without sharing information, as in Section \ref{sec :: equil}. 

To understand this result, consider the following. 
Let $W=\pi_1+\pi_2+CS = U(Q)$ denote welfare. Then, consider the expected value of the welfare benefits of information sharing, as follows 
\begin{align*}
\E_{w_1,w_2}[W(K,K^c)]&= \Pr\{L,H\} W_{L,H}(K,K^c)+\Pr\{H,L\} W_{H,L}(K,K^c)\nonumber\\
&\quad +\Pr\{L,L\} W_{L,L}(K,K^c)+\Pr\{H,H\} W_{H,H}(K,K^c)
\end{align*}
where the benefit to social welfare of sharing information between  wind producers at  state $\{w_1,w_2\}\in\{H,L\}^2$ is
\begin{align} \label{info_sharing}
W_{w_1,w_2}(K,K^c) \equiv W_{w_1,w_2}^K-W_{w_1,w_2}^{K^c}= Q_{w_1,w_2}^K - \frac{1}{2} \left(Q_{w_1,w_2}^K\right)^2- \bigg( Q_{w_1,w_2}^{K^c} - \frac{1}{2} \left(Q_{w_1,w_2}^{K^c}\right)^2 \bigg),
\end{align}
where $Q_{w_1,w_2}^K - \frac{1}{2} \left(Q_{w_1,w_2}^K\right)^2$ is  the social welfare at state $\{w_1,w_2\}$ when wind producers share their private information. Similarly, $Q_{w_1,w_2}^{K^c} - \frac{1}{2} \left(Q_{w_1,w_2}^{K^c}\right)^2$  denotes the social welfare when wind producers compete without sharing information. 

Equation \eqref{info_sharing} highlights the fact that information sharing has mixed effects on social welfare in different states. In particular, it increases total output quantity (and welfare) when only one producer is in the high state, but it decreases output quantity and social welfare when both producers are in the high state. However, since total production is relatively lower (and therefore $U'(Q))$ is relatively higher, when the producers are in opposite states, the net expected effect of information sharing is in favor of the benefits gained when producers are in opposite states. 

As $d$ increases, the benefits in the $\{H,L\}$ and $\{L,H\}$ states weakens and the costs incurred in the $\{H,H\}$ state increase, but the probability of being in the same state also decreases, so the proportional impact of the costs in state $\{H,H\}$ declines. Overall, information sharing improves social welfare for any $\beta, d \in (0,1)$ where Assumption \ref{assumption-11} is satisfied. 

Next, we consider the benefits of information sharing for producers' profits and show that in general they are not always aligned with the benefits for social welfare.

 \paragraph{{\color{Maroon}{\bf Is sharing information  beneficial for \emph{wind producers}?}}}
 While sharing information always improves social welfare, it is not always beneficial for wind producers. We show the answer depends on the extent of wind energy at the low state, which directly affects the variance in the aggregate output. When wind in the low state is sufficiently small, sharing information is extremely beneficial for a generator that is in its high state. As a result, ex-ante wind producers prefer to share information when $L$ is sufficiently small.   
 
 \begin{Proposition}\label{th-3}
There exists a threshold $L^*(d,\beta)$ that is increasing in $d$ (the dispersion between wind producers) and decreasing in the prior  $\beta$ so that sharing information is ex-ante beneficial for wind energy producers if only if $L<L^*(d,\beta)$.
\end{Proposition} 


Let $D(K,K^c)$ represent the change in profits due to information sharing. The result aims to characterize the sign of \eqref{eq:D}, which represents the expected value of the benefits of information sharing for producer profits.  
\begin{align}\label{eq:D}
\E_{w_1,w_2}[D(K,K^c)]&=\Pr\{L,H\} D_{L,H}(K,K^c)+\Pr\{H,L\} D_{H,L}(K,K^c)\nonumber\\
&\quad +\Pr\{L,L\} D_{L,L}(K,K^c)+\Pr\{H,H\} D_{H,H}(K,K^c)
\end{align}
The benefit of sharing information at  state $\{w_1,w_2\}\in\{H,L\}^2$ is
 \begin{align*}
 D_{w_1,w_2}(K,K^c)=\pi_{1_{w_1,w_2}}^K+\pi_{2_{w_1,w_2}}^K-\pi_{1_{w_1,w_2}}^{K^c}-\pi_{2_{w_1,w_2}}^{K^c},
 \end{align*}
  where $\pi_{i_{w_1,w_2}}^K$ denotes  $i$'s profit, $i\in \{1,2\}$, at state $\{w_1,w_2\}$ when firms  share their private information. Similarly, $\pi_{i_{w_1,w_2}}^{K^c}$  denotes $i$'s profit when firms compete with no information sharing. 

To understand the effects, first consider the benefits of information sharing in the $\{H,H\}$ and $\{L,L\}$ states. In the $\{H,H\}$ state sharing information is always beneficial because of improved coordination. In the $\{L,L\}$ state the benefit of sharing information is always zero because both firms produce  at the low level and thus there is no benefit for having better information. 

  \begin{figure}
 \begin{center}
\includegraphics[width=0.6\textwidth]{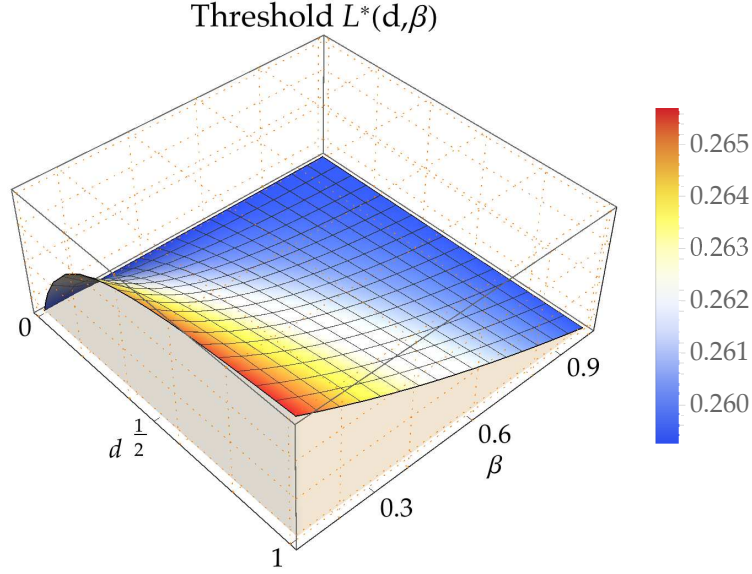}
\end{center}
\caption{This chart shows sharing information is ex-ante beneficial for wind energy producers if only if $L<L^*(d,\beta)$. Moreover, the threshold $L^*(d,\beta)$  is increasing in the dispersion $d$ and decreasing in the prior $\beta$. The gray area is for $L<L^*(d,\beta)$.}\label{fig:lstar}
\end{figure}

Now, suppose wind producer (WP) 1 is in the low state and WP 2 is in the high state. With information sharing, WP 2 achieves a best response to $w_1 = L$ and produces more energy compared to the case in which they do not share information. This benefits WP 2, because they achieve a best response based on improved information, but it hurts WP 1 because the price is reduced since WP 2 increases its output quantity. These effects favor information sharing when $L$ is relatively lower. Low $L$ improves the value of information sharing to WP 2 (because its overall adjustment is larger). Furthermore, low $L$ decreases the cost of information sharing to WP 1, because the price effect impacts a lower quantity of production since $L$ is small. The average benefit of information sharing, when producers are in different states $\{H,L\}$ or $\{L,H\}$, is decreasing in $L$. 

Overall, considering all the cases together  implies that the expected benefit of sharing information is controlled  by a threshold on the amount of wind energy in the low state. Therefore, sharing information is ex-ante beneficial for producers when wind energy at the low state is sufficiently small. This suggests that when the variance of wind availability for individual generators is high, they tend to benefit individually from information sharing; when the variance of their energy availability is low, information sharing reduces profits even though it improves social welfare. Furthermore, by increasing heterogeneity (i.e. the dispersion between the wind producers) this threshold increases, which incentivizes more wind energy producers to share their information.

\section{Conclusion} \label{Sec::conclusion}
This research links the heterogeneity in wind producer availability to social welfare in electricity markets with strategic behavior. It introduces the idea that the level of heterogeneity impacts the information signal quality that each wind producer's own availability provides; it shows that increasing heterogeneity decreases the strategic incentive of individual wind producers and increases their average output. This impact could become increasingly important as renewable energy penetration grows, especially because of the difficulties associated with monitoring market power when resource availability is not deterministic. 

The results show that increasing heterogeneity is generally beneficial because of its positive impacts on increasing diversification and on decreasing strategic curtailment. Some existing policies and subsidy models for wind energy, like state-level renewable portfolio standards, have been shown to impact the optimal investment locations for new projects; these effects should be further reviewed in light of these results. The research also highlights the benefits of publicly sharing high-quality real-time weather information, even when it is not in the best interest of producers. As such, policy makers should consider the potential benefits of improved public forecasting and publicly sharing real-time energy output data, understanding that welfare-improving policies may be opposed by electricity generators.       

 \newpage

\bibliographystyle{ormsv080}
\bibliography{dispersion_bib3_noURL}

\newpage 
 
%
%
\begin{appendix}


	
%

\section{Proofs omitted from main text} \label{Sec::Appendix}
%

\begin{proof}[{\color{Maroon}{\bf Proof of Proposition \ref{pro0}}}]
Since $P'<0, P''\leq 0$, firm $i$'s profit $\pi_i(q_i,q_j) = q_i P(q_i,q_j)$ is concave in $q_i$ regardless of the production $q_j$ by its competitor. Let firm $i$ be in the high state, i.e. $w_i=H$. By Assumption \ref{assumption-11}, $P(H)+HP'(H)<0$. Furthermore, $P(x)+xP'(x)$ is decreasing in $x$. Therefore, the resource availability in the high state, i.e. $q_i(H)=\phi \leq H$, does not bind. The optimal output $q_i(H) = \phi$ belongs to  $\arg \max_{q_i \in \mathbb{R}}\E_{w_j}[\pi_i|w_i=H]$.   Due to concavity of $\pi_i(q_i,q_j)$ in $q_i$, the first order optimality condition (the necessary and sufficient condition for optimality) implies that  $\phi$ should satisfy the following 
\begin{align}\label{star}
\Pr\{L|H\}\big[P(L+\phi)+\phi P'(\phi+L)\big]+\Pr\{H|H\}\big[P(2\phi)+\phi P'(2\phi)\big]=0,
\end{align}
given firm $j$ strategy is $q_j(w_j)=\min\{w_j,\phi\}$. 

Next, with the following Claims we show $\phi$ indeed satisfies \eqref{star} and verify that $q(L)=L$. Subsequently, we prove the symmetric equilibrium is unique. 

\noindent{{\bf Claim 1} \quad There exists a unique $\phi$  satisfying \eqref{star}. Moreover, $L<\phi<H$. }

\noindent{{\bf Proof}} \quad Let us define $f(x)\equiv \Pr\{L|H\}\big[P(L+x)+x P'(L+x)\big]+\Pr\{H|H\}\big[P(2x)+x P'(2x)\big]$. Taking a derivative of $f(x)$ with respect to $x$ implies
\begin{align*}
\Pr\{L|H\}\big[2P'(L+x)+ x P''(x+L)\big]+\Pr\{H|H\}\big[3P'(2x)+2 x P''(2x)\big]<0
\end{align*}
where the inequality follows by $P'<0, P''\leq 0$, $x \geq 0$. Thus, $f(x)$ is monotonically decreasing in $x$. 
Moreover, $f(L)=(P(2L)+LP'(2L))[\Pr\{H|H\}+\Pr\{L|H\}]>0$, which follows from Assumption \ref{assumption-11}. Furthermore, $f(H) <0 $ is bounded above by $0$:
\begin{align*}
f(H)&=\Pr\{L|H\}\big[P(L+H)+H P'(H+L)\big]+\Pr\{H|H\}\big[P(2H)+H P'(2H)\big]\\
&<(P(H)+HP'(H))[\Pr\{H|H\}+\Pr\{L|H\}]\\
&<0
\end{align*}
where the first inequality follows since $P(x+y)+xP'(x+y)$ is decreasing in $y$, and the second inequality follows by Assumption \ref{assumption-11}. Since $f(L)>0$, $f(H) < 0$, and $f'(x)<0$ thus there exists a unique $\phi$ for which $f(\phi)=0$, with $L < \phi < H$. 

\noindent{{\bf Claim 2}} \quad When $w_i=L$ then $q_i(L)=L$, given that  firm $j$'s strategy is $q_j(w_j)=\min\{w_j,\phi\}$.

\noindent{{\bf Proof}}\quad Let $g(x)\equiv \Pr\{H|L\}x P(\phi+x)+\Pr\{L|L\}x P(L+x)$.  We aim to show that $x=L$ maximizes $g(x)$ when $x\in [0,L]$.  The proof follows by contradiction. Suppose, by contrary, that the maximizer is $q_l<L$. Thus, first order optimality condition implies
that $q_l$ satisfies the following
\begin{align}\label{ql0}
\Pr\{H|L\}\big[P(\phi+q_l)+q_l P'(\phi+q_l)\big]+\Pr\{L|L\}\big[P(L+q_l)+q_l P'(L+q_l)\big]=0.
\end{align}
Since $P'<0, P''\leq 0$ thus $h(x)\equiv P(x+y)+xP'(x+y)$ is decreasing in $x\geq 0$,  for any $y$. Thus, $q_l<L$ implies that
\begin{align}\label{ql1}
h(q_l)|_{y=L}&=P(L+q_l)+q_lP'(L+q_l) \nonumber \\ 
&>P(2L)+LP'(2L) \nonumber \\
&>0
\end{align}
where the last inequality follows by Assumption \ref{assumption-11}. Furthermore, 
\begin{align}\label{ql2}
h(q_l)|_{y=\phi}=P(\phi+q_l)+q_lP'(\phi+q_l)&>\max \{P(L+\phi)+LP'(\phi+L), P(2\phi)+\phi P'(2\phi)\}\nonumber \\
&> \max \{P(L+\phi)+\phi P'(\phi+L), P(2\phi)+\phi P'(2\phi)\}\nonumber \\
& \geq 0
\end{align}
where the first inequality follows by $q_l<L<\phi$, the second follows since $\phi>L$ and $P'<0$, and the last inequality follows because \eqref{star} implies that it is impossible to have 
$0>\max \{P(L+\phi)+\phi P'(\phi+L), P(2\phi)+\phi P'(2\phi)\}$ (In fact, $P(L+\phi)+\phi P'(\phi+L)>0> P(2\phi)+\phi P'(2\phi)$).

Therefore, according to \eqref{ql1} and \eqref{ql2}, it is not possible for $q_l$ to satisfy  \eqref{ql0} when $q_l < L$. Thus, $q_i(L) \geq L$. But the constraint on production implies that $q_i(L) \leq L$, so $q_i(L) = L$.

\noindent{{\bf Claim 3}} \quad  The equilibrium described in Proposition \ref{pro0} is the unique symmetric equilibria.

The poof follows by contradiction. Suppose, by contrary, firm $j$ produces $q_j(L)=\tilde{L}$, where $\tilde{L}<L$. We show firm $i$ has incentive to deviate by producing more than $\tilde{L}$ in the low state. Suppose $q_j(H) = q_i(H)=\tilde{\phi}$; thus, $\tilde{\phi}$ (according to first order optimality condition) solves the following: 
\begin{align}\label{star-11}
\Pr\{L|H\}\big[P(\tilde{L}+\tilde{\phi})+\tilde{\phi} P'(\tilde{\phi}+\tilde{L})\big]+\Pr\{H|H\}\big[P(2\tilde{\phi})+\tilde{\phi} P'(2\tilde{\phi})\big]=0.
\end{align}
By following the arguments from Claim 1, there is indeed a unique $\tilde{\phi}$, where $\tilde{L}<\tilde{\phi}<H$, satisfying \eqref{star-11}. 

Now, let $w_i=L$. Then, evaluating  firm $i$'s marginal expected profit when $w_i=L$ and $q_i(L)=\tilde{L}$, given firm $j$'s strategy, implies
\begin{align}\label{contradiction}
\frac{\p}{\p q_i}\E_{w_j}[\pi_i(q_i,q_j)|w_i=L]|_{q_i=\tilde{L}}&=\Pr\{H|L\}\big[P(\tilde{\phi}+\tilde{L})+\tilde{L} P'(\tilde{\phi}+\tilde{L})\big]\nonumber \\
&\quad +\Pr\{L|L\}\big[P(2\tilde{L})+\tilde{L} P'(2\tilde{L})\big]\nonumber\\
&>0,
\end{align}
where the last inequality follows because of the following.  By Assumption \ref{assumption-11}, $P(2{L})+{L} P'(2{L})>0$. Also $\tilde{L}<L$ (by the above assumption) and $P(2x)+xP'(2x)$ is decreasing in $x\geq 0$. Thus $P(2\tilde{L})+\tilde{L} P'(2\tilde{L})>P(2{L})+{L} P'(2{L})>0$. 
In addition, since $P'<0,$ and $\tilde{\phi}>\tilde{L}$, thus $P(\tilde{\phi}+\tilde{L})+\tilde{L} P'(\tilde{\phi}+\tilde{L})>P(\tilde{\phi}+\tilde{L})+\tilde{\phi} P'(\tilde{\phi}+\tilde{L})>0$, where the last inequality holds because otherwise it is impossible for \eqref{star-11} to hold. Notice that
$P(\tilde{L}+\tilde{\phi})+\tilde{\phi} P'(\tilde{\phi}+\tilde{L})>P(2\tilde{\phi})+\tilde{\phi} P'(2\tilde{\phi})$. 

Finally, \eqref{contradiction} establishes a contradiction, because firm $i$ has incentive to deviate, and produce more than $\tilde{L}$ when $w_i = L$. This completes the proof.
\end{proof}

\begin{proof}[{\color{Maroon}{\bf Proof of Corollary \ref{pro1}}}]
Let $q_i(w_i)=\min\{w_i,\phi\}$. Consider $i=1$. The objective is to find $\phi$. Thus, writing the first order optimality condition implies that $\phi$ satisfies the following equality
\begin{align*}
\phi&=\frac{s-\E_{w_2}[q_2| w_1=H]}{2}\\
&=\frac{s-[L\Pr\{L|H\}+ \phi \Pr\{H|H\}]}{2}\\
&=\frac{s-[L(1-\frac{\beta}{\beta+d(1-\beta)})+\phi(\frac{\beta}{\beta+d(1-\beta)})]}{2}
\end{align*}
where $\Pr\{L|H\}=\Pr\{w_2=L|w_1=H\}=\frac{(1-\beta)d}{\beta+d(1-\beta)}$ and $\Pr\{H|H\}=\Pr\{w_2=H|w_1=H\}=\frac{\beta}{\beta+d(1-\beta)}$. The above equality gives $\phi \equiv \frac{s\beta+(s-L)(1-\beta)d}{3\beta+2(1-\beta)d}$,   completing the proof.
\end{proof}


\begin{proof}[{\color{Maroon}{\bf Proof of Lemma \ref{lem1}}}]
As shown in Proposition \ref{pro0}, production in the high state, i.e. $\phi$, solves 
\begin{align}\label{star-1}
\Pr\{L|H\}\big[P(L+\phi)+\phi P'(\phi+L)\big]+\Pr\{H|H\}\big[P(2\phi)+\phi P'(2\phi)\big]=0.
\end{align}
Furthermore, according to \eqref{prob}, since $\Pr\{L|H\}=\frac{d(1-\beta)}{\beta+d(1-\beta)}$ and $\Pr\{H|H\}=\frac{\beta}{\beta+d(1-\beta)}$ thus
\begin{align*}
\frac{\p}{\p d} \Pr\{L|H\}=\frac{\beta(1-\beta)}{(\beta+d(1-\beta))^2}&>0 \\
\frac{\p}{\p d} \Pr\{H|H\}=\frac{-\beta(1-\beta)}{(\beta+d(1-\beta))^2}&<0.
\end{align*}
Now, taking a derivative from \eqref{star-1} with respect to $d$ and taking into account that $\frac{\p}{\p d} \Pr\{H|H\}=-\frac{\p}{\p d} \Pr\{L|H\}<0$ gives
\begin{align*}
0=&\frac{\p \Pr\{L|H\}}{\p d}[P(L+\phi)+\phi P'(\phi+L)]+\Pr\{L|H\}\left [2\frac{\p \phi}{\p d}P'(L+\phi)+\phi \frac{\p \phi}{\p d}P''(\phi+L)\right]\\
&+\Pr\{H|H\}\left[3\frac{\p \phi}{\p d} P'(2\phi)+2\phi \frac{\p \phi}{\p d} P''(2\phi) \right]+\frac{\p \Pr\{H|H\}}{\p d}[P(2\phi)+\phi P'(2\phi)]\\
=&\frac{\p \phi}{\p d}\left\{\Pr\{L|H\}\big[2P'(L+\phi)+\phi P''(\phi+L)\big]+\Pr\{H|H\}\big[3P'(2\phi)+2\phi P''(2\phi)\big] \right\}\\
&+\frac{\p \Pr\{L|H\}}{\p d} [P(L+\phi)+\phi P'(L+\phi)-P(2\phi)-\phi P'(2\phi)].
\end{align*}
Therefore, 
\begin{align*}
\frac{\p \phi}{\p d}&=-\frac{\frac{\p \Pr\{L|H\}}{\p d} [P(L+\phi)+\phi P'(L+\phi)-P(2\phi)-\phi P'(2\phi)]}{\Pr\{L|H\}\big[2P'(L+\phi)+\phi P''(\phi+L)\big]+\Pr\{H|H\}\big[3P'(2\phi)+2\phi P''(2\phi)\big]}\\
&>0,
\end{align*}
where the inequality follows because: (i) $\frac{\p \Pr\{L|H\}}{\p d}>0$, (ii) $P'<0, P''\leq 0$, implying the denominator is negative, (iii) $P'<0, P''\leq 0$ and $L<\phi$, implying that $P(L+\phi)>P(2\phi), P'(L+\phi)\geq P'(2\phi)$.     
\end{proof}


\begin{proof}[{\color{Maroon}{\bf Proof of Example \ref{th-33}}}]
From \eqref{prob}, with prior probability $\Pr\{H\} = \beta$, we have 
	\begin{align}\label{chance}
		-\frac{\p \Pr\{L,L\}}{\p d}=-\frac{\p \Pr\{H,H\}}{\p d}=\frac{\p \Pr\{L,H\}}{\p d}\equiv {\color{Maroon}{\zeta}}=\frac{\beta^2(1-\beta)}{(\beta+d(1-\beta))^2}>0.
	\end{align}

The derivatives of the respective outcome probabilities are labeled as $\zeta$ and $-\zeta$ according to \eqref{chance}. By definition $\pi_i=q_iP(q_1+q_2)$. Therefore $\E_{w_1,w_2}[\pi_i]=\Pr\{L,H\}[LP(L+\phi)+\phi P(L+\phi)]+\Pr\{H,H\}\phi P(2\phi)+\Pr\{L,L\}LP(2L)$. Taking the derivative of average profit with respect to $d$ implies
 \begin{align*}
 \frac{\p}{\p d}\E_{w_1,w_2}[\pi_i]&=\underbrace{\zeta}_{>0} \underbrace{\big[LP(L+\phi)+\phi P(L+\phi)-LP(2L)-\phi P(2\phi) \big]}_{\equiv WD_{\pi}, \ \text{wind diversification}} \\ 
 &\quad+\underbrace{\frac{\p \phi}{\p d}}_{>0}\bigg\{\underbrace{\Pr\{L,H\}[LP'(\phi+L)+P'(\phi+L)\phi] +\Pr\{H,H\}[2P'(2\phi)\phi]}_{\equiv T_2, \ \text{effects of $d$ on price through its impact on strategic curtailment}}   \bigg\}  \\ 
 &\quad+\underbrace{\frac{\p \phi}{\p d}}_{>0}\bigg\{\underbrace{\Pr\{L,H\}[P(\phi+L)] +\Pr\{H,H\}[P(2\phi)]}_{\equiv T_3\ \text{the value of additional production due to reduced strategic curtailment}}   \bigg\}.  
 \end{align*}
$WD_{\pi}$ represents the effects of wind diversification, which is positive because
\begin{align*}
WD_{\pi}&=L[P(L+\phi)-P(2L)]+\phi[P(L+\phi)-P(2\phi)]\\
&>L\big[2P(L+\phi)-P(2L)-P(2\phi) \big] \\
&\geq 0
\end{align*}
where the first inequality follows because $\phi>L$ and $P(L+\phi)-P(2\phi)>0$ and the second inequality follows because of concavity in $P$, i.e. $P''\leq 0$. Thus, profit increases due to increased diversification. Note that unlike the effect of diversification on average price, which is inactive when $P''=0$, diversification improves profit even when the inverse demand curve is linear. 

Furthermore, the impacts of $d$ on strategic curtailment has two effects on profit, because reducing strategic curtailment lowers the average price but also increases the aggregate quantity; these impacts are labeled as 
$T_2$ and $T_3$, respectively. Since inverse demand is downward, i.e. $P'<0,$ the impact of increasing $d$ on markup through its effects on strategic curtailment is negative, i.e. $T_2<0$. The impact of reducing strategic curtailment on quantity is, expectedly, positive, i.e. $\frac{\p \phi}{\p d} > 0$ and $T_3>0$, because higher $d$ results in lower information and less extensive strategic curtailment. However, the overall impact of $d$, through its impacts on strategic curtailment, is to reduce price. This is because its effect on average price is greater than its effect on average quantity; i.e. $T_2+T_3<0$ because
\begin{align}
T_2+T_3&=\Pr\{H,L\}\big[P(\phi+L)+(L+\phi)P'(\phi+L) \big]+\Pr\{H,H\}\big[P(2\phi)+2\phi P'(2\phi) \big]\nonumber \\
&=\Pr\{H,L\}\big[P(\phi+L)+\phi P'(\phi+L)\big]+\Pr\{H,H\}\big[P(2\phi)+\phi P'(2\phi)\big]\nonumber \\
& \quad + \Pr\{H,L\} L P'(\phi+L)+ \Pr\{H,H\} \phi P'(2\phi) \nonumber \\ 
&= \Pr\{H,L\} L P'(\phi+L)+ \Pr\{H,H\} \phi P'(2\phi) \label{t3}\\
&<0  \label{t4}
\end{align} 
where \eqref{t3} follows from the first order condition \eqref{star}, and \eqref{t4} follows because inverse demand is downward, i.e. $P'<0$. Therefore, the effects od $d$ on diversification increase profits, and the effects of $d$ on strategic curtailment decrease profits. The overall impact of heterogeneity is ambiguous. Figure \ref{fig:profit} provides examples showing that profit can be increasing or decreasing in $d$.  
\end{proof}

\begin{proof}[{\color{Maroon}{\bf Proof of Proposition \ref{th-firms}}}]

By definition $\pi_i(w_1,w_2)=q_i(w_i)(s-q_1(w_1)-q_2(w_2))$ where $q_i(w_i)$ is explicitly given by Corollary  \ref{pro1}, for $w_i\in \{L,H\}$ and $i\in \{1,2\}$. The expected value of profit for producer $i$ is given \eqref{eq:30}.
\begin{align}\label{eq:30}
E_{w_1,w_2}[\pi_i]=\Pr\{L,H\} [\pi_i(L,H)+\pi_i(H,L)]+\Pr\{H,H\} \pi_i(H,H)+\Pr\{L,L\} \pi_i(L,L)
\end{align}
As before, from \eqref{prob}, $\Pr\{L,L\}=(1-\frac{d\beta}{\beta+d(1-\beta)})(1-\beta)$, $\Pr\{L,H\}= (1-\beta)\frac{d\beta}{\beta+d(1-\beta)}$, and $\Pr\{H,H\}=\beta\frac{\beta}{\beta+d(1-\beta)}$. In addition, 
\begin{align}\label{eq:31}
\pi_i(L,H)&=L(s-L-\phi)\\
\pi_i(H,L)&=\phi(s-L-\phi)\\
\pi_i(L,L)&=L(s-2L)\\
\pi_i(H,H)&=\phi(s-2\phi)\label{eq:34}
\end{align}
 where, as shown in Corollary \ref{pro1}, $\phi=\frac{s\beta+(s-L)(1-\beta)d}{3\beta+2(1-\beta)d}$. By plugging \eqref{eq:31}-\eqref{eq:34} into \eqref{eq:30}, the total (ex-ante) wind producers' surplus becomes
\begin{align}\label{PS}
\E_{w_1,w_2}[\pi_i]=\frac{\beta}{4}+L(1-2\beta)+L^2(\frac{15}{4}\beta-2)-\frac{\beta^2(s-3L)(s-4L)}{2(3\beta+2d(1-\beta))}+\frac{\beta^3(s-3L)^2}{4(3\beta+2d(1-\beta))^2}.
\end{align} 
Next, we characterize how $d$ (extend of heterogeneity) affects profits. Taking a derivative with respect to $d$ from \eqref{PS} implies 
\begin{align}\label{eq:t1}
\frac{\p }{\p d}\E_{w_1,w_2}[\pi_i]&=\frac{-\beta^3(s-3L)^2(1-\beta)}{(3\beta+2d(1-\beta))^3}+\frac{\beta^2(1-\beta)(s-3L)(s-4L)}{(3\beta+2d(1-\beta))^2}\nonumber \\
&=\frac{\beta^2(1-\beta)(s-3L)}{(3\beta+2d(1-\beta))^2}\ \left[\frac{-\beta(s-3L)}{3\beta+2d(1-\beta)}+s-4L\right]\nonumber\\
&=\underbrace{\frac{\beta^2(1-\beta)(s-3L)}{(3\beta+2d(1-\beta))^3}}_{>0}\ \left[\beta(2s-9L)+d(1-\beta)(2s-8L)\right].
\end{align} 
From the last equality we obtain: If $L<\frac{2s}{9}\equiv L_1$, then $2s-9L>0$ and $2s-8L>0$; thus, $\frac{\p }{\p d}\E_{w_1,w_2}[\pi_i]>0$ and, consequently, $\arg \max_{d\in [0,1]} \ \E_{w_1,w_2}[\pi_i]= 1$. If $L>\frac{2s}{8}\equiv L_2$ then $2s-8L<0$ and $2s-9L<0$; thus $\frac{\p }{\p d}\E_{w_1,w_2}[\pi_i]<0$ and, consequently, $\arg \max_{d\in [0,1]} \ \E_{w_1,w_2}[\pi_i]= 0.$

 
In sum, \eqref{eq:t1} implies that if $L<L_1$ then $\frac{\p }{\p d}\E_{w_1,w_2}[\pi_i]>0$ and thus $\max_{d\in [0,1]} \E_{w_1,w_2}[\pi_i]$ happens at $d=1$. Similarly, if $L>L_2$ then $\frac{\p }{\p d}\E_{w_1,w_2}[\pi_i]<0$ and thus $\max_{d\in [0,1]} \E_{w_1,w_2}[\pi_i]$ occurs  at $d=0$.
For the sake of completeness, we further note that $\arg \max_{d\in [0,1]} \ \E_{w_1,w_2}[\pi_i]\in\{0,1\}$ for \emph{any} $L<\frac{s}{3}$.\footnote{This is because any interior $\tilde{d} \in (0,1)$ such that $\frac{\p }{\p d}\E_{w_1,w_2}[\pi_i]=0$ implies that $\frac{\p }{\p d}\E_{w_1,w_2}[\pi_i] \big|_{d>\tilde{d}}>0$. Thus any $d \in \{0,1\}$ for any $d$ that maximizes profits.}
\end{proof}

\begin{proof}[{\color{Maroon}{\bf Proof of Proposition \ref{th-gen_SC}}}]
		Let $\phi^d$ be the equilibrium high state production that for each $i \in \{1,...,N+1\}$ satisfies 
		\begin{equation} \label{multFOC_d}
		\mathbb{E}_{S^d_{-i}}[P(\phi^d+S^d_{-i}\phi^d + (N-S^d_{-i})L) + \phi^d P'(\phi^d + S^d_{-i}\phi^d + (N-S^d_{-i})L)|w_i=H] = 0.
		\end{equation}		
		The left hand side of \eqref{multFOC_d} is decreasing in $\phi^d$. Its derivative with respect to $\phi^d$ is 		
		\begin{equation} \label{multFOC_dd}
		\mathbb{E}_{S^d_{-i}}[(2+ S^d_{-i}) P'(\phi^d+S^d_{-i}\phi^d + (N-S^d_{-i})L) + \phi^d P''(\phi^d + S^d_{-i}\phi^d + (N-S^d_{-i})L)|w_i=H] < 0
		\end{equation}
		which follows because $2 + S^d_{-i} >0$, $\phi^d >0$, $P'<0$, and $P'' \leq 0$. 
		
		Now, consider $d' > d$ and assume towards a contradiction that $\phi^{d'} < \phi^d$, where $\phi^d$ and $\phi^{d'}$ satisfy the first order condition in \eqref{multFOC_d}, with the expectations taken according to their respective random variables $S^d_{-i}$ and $S^{d'}_{-i}$: 
		\begin{align*} \label{dprime_fails}
		0 & = \mathbb{E}_{S^d_{-i}}[P(\phi^d+S^d_{-i}\phi^d + (N-S^d_{-i})L) + \phi^d P'(\phi^d + S^d_{-i}\phi^d + (N-S^d_{-i})L)|w_i=H] \\
		& < \mathbb{E}_{S^d_{-i}}[P(\phi^{d'}+S^d_{-i}\phi^{d'} + (N-S^d_{-i})L) + \phi^{d'} P'(\phi^{d'} + S^d_{-i}\phi^{d'} + (N-S^d_{-i})L)|w_i=H] \\
		& \leq \mathbb{E}_{S^{d'}_{-i}}[P(\phi^{d'}+S^{d'}_{-i}\phi^{d'} + (N-S^{d'}_{-i})L) + \phi^{d'} P'(\phi^{d'} + S^{d'}_{-i}\phi^{d'} + (N-S^{d'}_{-i})L)|w_i=H].
		\end{align*}  
		The first inequality is due to \eqref{multFOC_dd}, with $\phi^{d'} < \phi^d$. The second inequality is due to Assumption \ref{assumption-FOSD} with $F(x) = P(\phi^{d'}+x \phi^{d'} + (N-x)L) + \phi^{d'} P'(\phi^{d'} + x\phi^{d'} + (N-x)L)$ decreasing in $x$. But the result implies that $\phi^{d'}$ does not satisfy the first order condition of the equilibrium, so we have a contradiction. Therefore, $\phi^d$ is (weakly) increasing in $d$. 
\end{proof}

\begin{proof}[{\color{Maroon}{\bf Proof of Proposition \ref{th-gen_welfare}}}]
	
By definition, $W = U(Q)$ where $Q = \sum_{i=1}^{N+1} q_i$. Note that $U'(Q) = P(Q) \geq 0$ for any equilibrium $Q$, and $P' < 0$, $\phi > 0$. Furthermore, note that we can write $Q$ as a function of $d$ given any realization of availability, $Q(d) = \sum_{i} s_i (\phi^d-L) + (N+1) L = (\phi^d - L)S + (N+1) L$, which is increasing and linear in $S$. This implies that the expectation $\mathbb{E}[W]$, which is taken over the random states of all of the producers, is fully defined by the probability distribution of $S$. Let $\vec{s}$ be the random vector of states of each of the producers, i.e. $\vec{s} = [s_1,s_2,...s_{N+1}]$. Then for all $i$, $s_i$, with $S = \sum_{i} s_i$, we  have that $\mathbb{E}_{\vec{s}}[W(Q(d))] = \mathbb{E}_{S}[W(Q(d))]$. Then for $d' > d$: 
	\begin{align}
	\mathbb{E}_{\vec{s}^{d'}}[W] - \mathbb{E}_{\vec{s}^d}[W] & = \mathbb{E}_{S^{d'}}[W(Q(d'))] - \mathbb{E}_{S^d}[W(Q(d))]  \nonumber\\
	& = \sum_{k=1}^{N+1} U(Q(d')) \Pr\{S^{d'} = k\} - \sum_{k=1}^{N+1} U(Q(d)) \Pr\{S^d = k\}  \nonumber\\
	& \geq \sum_{k=1}^{N+1} U(Q(d)) (\Pr\{S^{d'} = k\} - \Pr\{S^d = k\})   \nonumber\\
	& = \mathbb{E}_{\vec{s}^{d'}}[U(Q(d))] - \mathbb{E}_{\vec{s}^d}[U(Q(d))] \geq 0. \nonumber
	\end{align}  
	The first line is because $S$ provides equivalent information for the expectation as explained above. The second line is an expansion of the expectations for the discrete random variables, and the third line is because of Proposition \ref{th-gen_SC} with $Q(d)$ increasing in $d$ and $U(Q)$ increasing in $Q$. The fourth line rewrites the third as a difference of expectations, which is non-negative by the definition of second-order stochastic dominance with $U$ increasing and concave in $S$ since $U$ is increasing and concave and $Q(S) = (\phi - L)S + (N+1) L $ is linear and increasing in $S$. 
\end{proof}

\begin{proof} [{\color{Maroon}{\bf Proof of Example \ref{prop :: trad_gen eq}}}]
	For the linear inverse demand, conditions \eqref{A} and \eqref{B} are represented by \eqref{Alin} and \eqref{Blin}.
	\begin{equation} \label{Alin}
	\Pr\{L|H\} (s - L - 2 \phi - x) + \Pr\{H|H\} (s - 3 \phi - x) = 0 
	\end{equation} 
	\begin{equation} \label{Blin}
	\Pr\{L,L\} (s - 2 L - 2 x) + 2 \Pr\{L,H\} (s - L - \phi - 2 x) + \Pr\{H,H\} (s - 2 \phi - 2 x) - c = 0 
	\end{equation} 
	Under linear inverse demand, \eqref{Alin} is the first order condition for the wind generators, and \eqref{Blin} is the first order condition for the traditional generator with constant marginal cost $c$. Equation \eqref{Alin} allows us to write the expression for $\phi$ in equilibrium, analogous to the result from Corollary \ref{pro1} but including the effect of $x$.
	\begin{equation} \label{phiw}
	\phi = \frac{(s-x)\beta + (s-x-L)d(1-\beta)}{3\beta + 2d(1-\beta)} = \frac{s\beta + (s-L)d(1-\beta) - x(\beta + d(1-\beta))}{3\beta + 2d(1-\beta)} 
	\end{equation} 
	We can rearrange the generator's first order condition \eqref{Blin} to obtain equation \eqref{wphi}. 
	By the assumption, with $\phi < H$, the traditional generator chooses to participate in the market, i.e. \eqref{wphi} is solved by some $x \geq 0$. Combining equations \eqref{phiw} and \eqref{wphi}, we have the result in \eqref{phi}. The uniqueness of the symmetric\footnote{The equilibrium is symmetric in the sense that the wind producers have identical strategies; the traditional producer has a different objective and a different strategy from the wind producers.} equilibrium is clear from the fact that the result for $\phi$ in \eqref{phi} does not depend on $x$. 
\end{proof}

\begin{proof} [{\color{Maroon}{\bf Proof of Proposition \ref{prop :: trad_gen welfare}}}]
	Average output is given by 
	\begin{align*}
	\mathbb{E}[q_1(w_1)+q_2(w_2)+x] & = 2 \Pr\{L,H\}(L+\phi + x) + \Pr\{L,L\}(2L + x) + \Pr\{H,H\}(2\phi + x)  \\ &= x + 2 \beta \phi + 2 (1-\beta) L.
	\end{align*}
	Taking the derivative with respect to $d$, we have that
	\begin{equation*}
	\frac{\partial \mathbb{E}[q_1(w_1)+q_2(w_2)+x]}{\partial d} = \beta \frac{\partial \phi}{\partial d} > 0,
	\end{equation*} 
	due to the linearity of expectation and because $\frac{\partial \mathbf{E}[q_i(w_i)]}{\partial d} = \beta \frac{\partial \phi}{\partial d}$, for $i \in \{1,2\}$, and $\frac{\partial x}{\partial d} = -\beta \frac{\partial \phi}{\partial d}$.	Now, average welfare is given by  
	\begin{equation*}
	\mathbb{E}_{w_1,w_2}[W] = 2 \Pr\{L,H\}U(L+\phi + x) + \Pr\{L,L\}U(2L + x) + \Pr\{H,H\}U(2\phi + x) - c(x).
	\end{equation*}
	Taking the derivative with respect to $d$,
	\begin{align*}
	\begin{split}
	\frac{\partial \mathbb{E}_{w_1,w_2}[W]}{\partial d} & = \zeta(2U(L+\phi + x) - U(2L + x) - U(2\phi + x)) + \frac{\partial x}{\partial d} \Pr\{L,L\}P(2L + x) \\ & + 2 (\frac{\partial \phi}{\partial d} + \frac{\partial x}{\partial d})  \Pr\{L,H\}P(L+\phi + x) +  (2 \frac{\partial \phi}{\partial d} + \frac{\partial x}{\partial d})\Pr\{H,H\}P(2\phi + x) - c \frac{\partial x}{\partial d} 
	\end{split} \nonumber \\
	\begin{split} &= \Gamma + 2 \frac{\partial \phi}{\partial d}  \Pr\{L,H\}P(L+\phi + x) +  2 \frac{\partial \phi}{\partial d} \Pr\{H,H\}P(2\phi + x) + \frac{\partial x}{\partial d} x  \end{split} \nonumber \\
	\begin{split}
	& = \Gamma + \beta \frac{\partial \phi}{\partial d} (2 \phi - x) \geq \Gamma \geq 0.
	\end{split}
	\end{align*}
	
	\highlighta{The second and third lines replace the first term with $\Gamma = \zeta(2U(L+\phi + x) - U(2L + x) - U(2\phi + x))$ to concatenate the expression; $\Gamma$ is the impact of wind diversification on welfare, and it is weakly positive due to the concavity of $U$, as explained in Proposition \ref{prop2}. The second equality uses the first order condition from \eqref{Blin}. The third equality uses the conditional probabilities $\Pr\{L,H\} = \Pr\{L | H\} \beta$ and $\Pr\{H,H\} = \Pr\{H | H\} \beta$, along with the first order condition from \eqref{Alin}. The expression $2 \phi - x$ is minimized when $c = 0$ and when $d = 0$. Therefore, by using  \eqref{phi} and \eqref{wphi},  with $c,d = 0$, we confirm that $2 \phi - x \geq 0$. This fact and $\Gamma \geq 0$ establish the inequalities in the final line and complete the proof. }
\end{proof}


\begin{proof} [{\color{Maroon}{\bf Proof of Proposition \ref{prop :: trad_gen price}}}]
	Average price can be expressed as
	\begin{equation*}
	\mathbb{E}_{w_1,w_2}[P] = 2 \Pr\{L,H\}P(L+\phi + x) + \Pr\{L,L\}P(2L + x) + \Pr\{H,H\}P(2\phi + x).
	\end{equation*}
	Taking the derivative with respect to $d$ gives: 
	\begin{equation}
	\begin{aligned}
	\frac{\partial \mathbb{E}_{w_1,w_2}[P]}{\partial d} &= \zeta (2P(L+\phi + x) - P(2 L + x) - P(2 \phi + x)) - 2 \frac{\partial \phi}{\partial d} (\Pr\{L,H\} + \Pr\{H,H\}) - \frac{\partial x}{\partial d} \nonumber \\
	&= \zeta 0 - 2\beta \frac{\partial \phi}{\partial d} + \beta \frac{\partial \phi}{\partial d} = -\beta \frac{\partial \phi}{\partial d}.  \nonumber
	\end{aligned}
	\end{equation}
	This is due to the fact that $P(x)$ represents linear inverse demand, so the first term sums to 0 and so $\forall x$ $P'(x)=-1$, and also due to the fact that $\Pr\{L,H\} + \Pr\{H,H\} = \beta$. This completes the proof. As in the two-producer case, for a linear inverse demand curve, average price is monotonically decreasing in $d$.
	
	
	%
	
\end{proof}


\begin{proof} [{\color{Maroon}{\bf Proof of Proposition \ref{collusion_poss}}}]
	We seek to prove that there always exists a suitable $t$ that satisfies \eqref{IC}, \eqref{IRH}, and \eqref{IRL} when $\gamma=0$. Let $\gamma=0$ by assumption. Rearranging the IC condition gives
	\begin{equation}\label{IC2}
	t \leq \frac{1}{2}\Pr\{H|H\} + \Pr\{L|H\}(1- \frac{\pi_L}{\pi_M}).
	\end{equation}
	Rearranging the IR-H condition provides another upper bound on $t$:
	\begin{equation} \label{IRH2}
	t \leq 1+ \frac{\beta}{2 d(1-\beta)} - \frac{\beta}{d(1-\beta)} \frac{\phi (s-2\phi)}{\pi_M} - \frac{\phi (s-\phi - L)}{\pi_M} - \frac{\gamma}{\Pr\{L|H\} \pi_M}.
	\end{equation}
	Rearranging the IR-L condition provides a lower bound for $t$:
	\begin{equation} \label{IRL2}
	t \geq \frac{L(s-\phi - L)}{\pi_M} + \frac{\gamma}{\Pr\{H|L\} \pi_M}.
	\end{equation}
	The proof follows by showing that the lower bound for $t$, the right-hand side (RHS) of \eqref{IRL2} is always less than or equal to the upper bounds for $t$ from the RHS of \eqref{IC2} and \eqref{IRH2} when $\gamma = 0$. Thus, there is always a nonempty feasible subset of $\mathbb{R}$ from which a transfer $t$ can be selected that satisfies the criteria for collusion.  
	 
	First, with $\gamma = 0$, the RHS of \eqref{IRL2} is always less than the RHS of \eqref{IC2}. The RHS of \eqref{IRL2},
	\begin{equation*}
	\frac{L(s-\phi - L)}{\pi_M} < \frac{L(s-2L)}{\pi_M} \leq 1/2,
	\end{equation*} 
	where the first inequality is due to $\phi > L$ and the second is because $L(s-2L)$ is maximized at $\frac{s^2}{8}$ when $L = s/4$, and because $\pi_M = \frac{s^2}{4}$. Furthermore, from the RHS of \eqref{IC2},
	\begin{equation*}
	\frac{1}{2}\Pr\{H|H\} + \Pr\{L|H\}(1- \frac{\pi_L}{\pi_M}) \geq \frac{1}{2}\Pr\{H|H\} + \Pr\{L|H\}(1- \frac{4}{8}) = 1/2
	\end{equation*} 
	where the inequality is because $\pi_L \leq \frac{s^2}{8}$ and the equality is because the expression is a weighted probabilistic sum of two values equal to $1/2$. Therefore, the lower bound on $t$, \eqref{IRL2} is at most $\frac{1}{2}$, and one upper bound on $t$, \eqref{IC2}, is at least $\frac{1}{2}$.
	
	Now, it remains to be shown that the RHS of \eqref{IRH2} (the other upper bound on $t$) is at least as great as the RHS of \eqref{IRL2}. Equivalently, their difference $T$ is greater than or equal to zero:
	\begin{equation*} 
	\begin{aligned} 
	T & = 1+ \frac{\beta}{2 d(1-\beta)} - \frac{\beta}{d(1-\beta)} \frac{\phi (s-2\phi)}{\pi_M} - \frac{\phi (s-\phi - L)}{\pi_M} - \frac{L(s-\phi - L)}{\pi_M} \\ 
	& = \frac{\beta}{d(1-\beta)}(\frac{1}{2} - \frac{\phi (s-2\phi)}{\pi_M}) + \frac{\pi_M - (\phi+L)(s-\phi-L)}{\pi_M} \geq 0.
	\end{aligned} 
	\end{equation*}
	\highlighta{The second line is a rearrangement of the first. The first term in the second line is positive because $\pi_M = \max_\phi 2 \phi (s - 2 \phi)$ by the definition of monopoly profits; therefore, $\frac{\phi (s - 2 \phi)}{\pi_M} \leq \frac{1}{2}$. The second term is positive since $\pi_M = \max_{\phi,L} (\phi + L)(s-\phi - L)$, again by the definition of Monopoly profits. Therefore, none of the aforementioned constraints contradict, and there always exists some $t$ that allows the producers to collude.}     
\end{proof}

%

\begin{proof} [{\color{Maroon}{\bf Proof of Proposition \ref{th-6}}}]

\quad Let $BR_i(\zeta)=\frac{1- \zeta}{2}$ denote $i$'s best reply  when $q_j=\zeta$. We aim to characterize the expected value of information sharing, which is given by 
\begin{align}\label{eq:W}
\E_{w_1,w_2}[W(K,K^c)]&=\Pr\{L,H\} W_{L,H}(K,K^c)+\Pr\{H,L\} W_{H,L}(K,K^c)\nonumber\\
&\quad +\Pr\{L,L\} W_{L,L}(K,K^c)+\Pr\{H,H\} W_{H,H}(K,K^c),
\end{align}
where (according to \eqref{prob}), $\Pr\{L,L\}=(1-\beta)(1-\frac{d\beta}{\beta+d(1-\beta)})$, $\Pr\{H,L\}=\Pr\{L,H\}= (1-\beta)\frac{d\beta}{\beta+d(1-\beta)}$, and $\Pr\{H,H\}=\beta\frac{\beta}{\beta+d(1-\beta)}$.
 
 The benefit of cooperation/sharing information at  state $\{w_1,w_2\}\in\{H,L\}^2$, denoted by $W_{w_1,w_2}(K,K^c)$, is
given by \eqref{info_sharing},
and $Q_{w_1,w_2}^K$ denotes total output at state $(w_1,w_2)$ when firms cooperate and share their private information. Similarly, $Q_{w_1,w_2}^{K^c}$  denotes total output when firms compete with no shared information. We consider four separate cases as follows:

\noindent{\bf{Case 1: $\{L,H\}$}.} \quad In this case WP 1 is in the low state and WP 2 is in the high state. Information sharing increases total output because WP 2 can produce more energy, knowing for certain that WP 1 can only produce $L$ units rather than $\mathbb{E}_{w_1}[q_1 | w_2 = H] = L\Pr\{L|H\}+ \phi \Pr\{H|H\} > L$. Therefore, 
\begin{equation*}
	Q_{L,H}^K = L + BR_2(L)= L + \frac{1- L}{2} = \frac{1 + L}{2}
\end{equation*}  
whereas $Q_{L,H}^{K^c} = L + \phi$. 

\noindent{\bf{Case 2: $\{H,L\}$}} \quad This case by symmetry is identical to Case 1. 
\medskip
\medskip

\noindent{\bf{Case 3: $\{H,H\}$}} 
\quad In this case, information sharing reduces total output because the producers learn that the opposing producers have the ability to produce at the Cournot level, since information sharing eliminates the possibility that the other producer is in the low state (which causes them to overproduce, given that the other producer is in the high state). Under information sharing, both producer produce at the Cournot level. Therefore, $Q_{L,H}^K = 2 q_C = \frac{2}{3}$. In the absence of information sharing $Q_{L,H}^{K^c} = 2 \phi$. 

\medskip
\medskip

\noindent{\bf{Case 3: $\{L,L\}$}}   
\quad In this case WP 1 and WP 2 are both in the low state. Thus there is no difference between cooperation and competition since both produce at the $L$ level, meaning that $W_{L,L}(K,K^c)=0.$ 

Plugging these results into \eqref{eq:W} and \eqref{info_sharing}, we have that
\begin{align*}
	\E_{w_1,w_2}[W(K,K^c)]&= \frac{d\beta (1-\beta) }{\beta+d(1-\beta)} \bigg( 2 \frac{1+L}{2} - \left( \frac{1+L}{2}\right)^2 - 2 \left(L + \phi \right) + \left(L + \phi \right)^2 \bigg) \\ 
	& + \frac{\beta^2 }{\beta+d(1-\beta)} \bigg(\frac{2}{3} - \frac{1}{2} \left(\frac{2}{3}\right)^2 - 2\phi + \frac{1}{2} \left(2 \phi \right)^2 \bigg).  
\end{align*}
By rearranging the above equation, we have that
\begin{equation} \label{W_expanded} 
\E_{w_1,w_2}[W(K,K^c)]= \Gamma(\beta,d,L) \bigg( 39\beta + 28 d (1-\beta) - 60 L d (1-\beta) - 81 \beta L \bigg).  
\end{equation} 
The common factor 
\begin{equation*}
	\Gamma(\beta,d,L) = \frac{\beta^2 d (1-3 L) (1-\beta) }{36\left(\beta+d(1-\beta)\right)\left(3 \beta + 2 d (1-\beta)\right)^2}
\end{equation*}
is positive because $L < 1/3$ (equivalently, $L < s/3$ for general $s$) by Assumption \ref{assumption-11}, $(1-\beta) \in (0,1)$. Similarly, since $L < 1/3$, the additive terms of \eqref{W_expanded} 
\begin{align*}
	39\beta + 28 d (1-\beta) - 60 L d (1-\beta) - 81 \beta L & > 	39\beta + 28 d (1-\beta) - 20 d (1-\beta) - 27 \beta \\ &= 12 \beta + 8 d (1-\beta) > 0.
\end{align*}
The social welfare benefit of information sharing $\E_{w_1,w_2}[W(K,K^c)]$ is the product of two positive terms, and therefore $E_{w_1,w_2}[W(K,K^c)]>0$. 
\end{proof}

\begin{proof}[{\color{Maroon}{\bf Proof of Proposition \ref{th-3}}}]
 \quad Let $BR_i(\zeta)=\frac{1- \zeta}{2}$ denote $i$'s best reply  when $q_j=\zeta$. We aim to characterize the following 
\begin{align*}\label{eq:D:pr}
\E_{w_1,w_2}[D(K,K^c)]&=\Pr\{L,H\} D_{L,H}(K,K^c)+\Pr\{H,L\} D_{H,L}(K,K^c)\nonumber\\
&\quad +\Pr\{L,L\} D_{L,L}(K,K^c)+\Pr\{H,H\} D_{H,H}(K,K^c)
\end{align*}
 where the benefit of cooperation/sharing information at  state $\{w_1,w_2\}\in\{H,L\}^2$ is
 \begin{equation*}
 D_{w_1,w_2}(K,K^c)=\pi_{1_{w_1,w_2}}^K+\pi_{2_{w_1,w_2}}^K-\pi_{1_{w_1,w_2}}^{K^c}-\pi_{2_{w_1,w_2}}^{K^c}
 \end{equation*}
  and $\pi_{i_{w_1,w_2}}^K$ denotes  $i$'s profit at state $(w_1,w_2)$ when firms cooperate and share their private information. Similarly, $\pi_{i_{w_1,w_2}}^{K^c}$  denotes $i$'s profit when firms compete with no information sharing (no cooperation). We consider four possible cases separately as follows.

\noindent{\bf{Case 1: $\{L,H\}$}.} \quad 
  In this case WP 1 is in the low state and WP 2 is in the high state. Thus, information sharing is highly beneficial for WP 2 (and detrimental for WP 1). This is due to the fact that 
  producer 2 will strategically overproduce and thus price goes down, hurting producer 1. 

This overproduction is beneficial for producer 2, even though it results in a decrease in the equilibrium  price. In this case, cooperation is beneficial for WP 2 and detrimental for WP 1, compared to competition with no  information sharing. It is intuitive and important to note that the extra  benefit to WP 2 from information sharing is {\bf particularly high}   when $L$  is small.   More precisely, 
  \begin{align*}
  \pi_{1_{L,H}}^K&=L[1-(L+BR_2(L))]\\
  \pi_{2_{L,H}}^K&=BR_2(L)[1-(L+BR_2(L))],
  \end{align*}
  where $BR_2(L)=\frac{1- L}{2}$. 
  With no cooperation, each WP supplies according to the original equilibrium. Thus,
    \begin{align*}
  \pi_{1_{L,H}}^{K^c}&=L[1-(L+\phi)]\\
  \pi_{2_{L,H}}^{K^c}&=\phi [1-(L+\phi)]
  \end{align*}
  with $\phi=\frac{\beta+(1-L)(1-\beta)d}{3\beta+2(1-\beta)d}$ by Corollary \ref{pro1}.  With algebra we can show
    \begin{align*}
  D_{H,L}(K,K^c)&=\pi_{1_{H,L}}^K+\pi_{2_{H,L}}^K-\pi_{1_{H,L}}^{K^c}-\pi_{2_{H,L}}^{K^c}\\
  &=\left(\frac{\beta(1-3L)}{2(2d(1-\beta)-3\beta)}\right)^2- L\left(\frac{\beta(1-3L)}{2(2d(1-\beta)-3\beta)}\right).
  \end{align*}
  
  \medskip
  \medskip
  
  \noindent{\bf{Case 2: $\{H,L\}$}} \quad This case by symmetry is similar to Case 1. Thus, 
    \begin{align*}
  D_{H,L}(K,K^c)&=\pi_{1_{H,L}}^K+\pi_{2_{H,L}}^K-\pi_{1_{H,L}}^{K^c}-\pi_{2_{H,L}}^{K^c}\\
  &=\left(\frac{\beta(1-3L)}{2(2d(1-\beta)-3\beta)}\right)^2- L\left(\frac{\beta(1-3L)}{2(2d(1-\beta)-3\beta)}\right).
  \end{align*}
  
  \medskip
  \medskip
  
  \noindent{\bf{Case 3: $\{H,H\}$}} \quad  In this case WP 1 and WP 2 are both in the high state. Thus, cooperation is highly beneficial for both of them because by reducing uncertainty  they both optimally coordinate and produce at the corresponding Cournot level, i.e.  $q_C=\frac{1}{3}$. Thus the profit with information sharing is characterized as follows:
  \begin{align*}
  \pi_{1_{H,H}}^K=\pi_{2_{H,H}}^K=q_C[1-2q_C]
  \end{align*}
With no cooperation each WP best replies to her belief; thus,
  \begin{align*}
  \pi_{1_{H,H}}^{K^c}= \pi_{1_{H,H}}^{K^c}= \phi [1-2\phi]
  \end{align*}
  where (as specified above) $\phi=\frac{\beta+(1-L)(1-\beta)d}{3\beta+2(1-\beta)d}$. Using algebra, we can show
    \begin{align*}
  D_{H,H}(K,K^c)&=\pi_{1_{H,H}}^K+\pi_{2_{H,H}}^K-\pi_{1_{H,H}}^{K^c}-\pi_{2_{H,L}}^{K^c}\\
  &=2q_C[1-2q_C]-2 \phi [1-2\phi]\\
  &=\frac{2}{9}+\frac{\left(\frac{2(\beta+d(1-L)(1-\beta))}{3\beta+2d(1-\beta)}-1\right)(\beta+d(1-L)(1-\beta))}{3\beta+2d(1-\beta)}\\
  &\geq 0
  \end{align*}
where the last inequality follows because $\phi>q_C=\frac{1}{3}$ (because of the overproduction of each WP producing $\phi$ in this state due to uncertainty and mis-coordination).\footnote{Note that $f(x)=x(1-2x)$ is concave in $x$ and is maximized at $x=\frac{1}{4}$, thus $f(q_C)>f(\phi)$, because $\phi>q_C=\frac{1}{3}>\frac{1}{4}$. }
 \medskip
  \medskip
  
  \noindent{\bf{Case 4: $\{L,L\}$}} \quad  In this case WP 1 and WP 2 are both in the Low state. Thus there is no difference between cooperation and competition since both produce at the $L$ level, meaning that $D_{L,L}(K,K^c)=0.$ 
   
   \medskip
  \medskip
  
  \noindent Plugging the results of the above cases into \eqref{eq:D} implies that
  \begin{align*}\label{eq:D1}
 \E_{w_1,w_2}[D(K,K^c)]= \frac{\beta^2d(1-3L)(1-\beta)}{(3\beta+2d(1-\beta))^2(\beta+d(1-\beta))}\big[21\beta+16d(1-\beta)-L(81\beta+60d(1-\beta))\big].
  \end{align*}
  As a result there exists a \emph{unique} $L^*(d,\beta)$ such that 
  $\E_{w_1,w_2}[D(K,K^c)]>0$ if and only if 
  \begin{equation*}
  L<L^*(d,\beta)=\frac{21\beta+16d(1-\beta)}{81\beta+60d(1-\beta)}.
  \end{equation*}
  In the above expression, note that $L^*(d,\beta)<\frac{1}{3}$.
\end{proof}

\end{appendix}

\end{document}